\documentclass[acmsmall]{acmart}

\usepackage{tcolorbox}
\usepackage{xcolor}
\usepackage{amsmath}

\usepackage{enumitem}

\usepackage{setspace}
\usepackage{algpseudocode}
\usepackage{xcolor}
\usepackage{booktabs}
\usepackage{multirow}
\usepackage{pifont}
\usepackage{pgfplots}
\usepackage{makecell}
\usepackage{wrapfig}
\usepackage{overpic}
\usepackage{subcaption}

\usetikzlibrary{positioning, arrows.meta, calc, shadows}
\usepackage{tikz}
\usepackage{listings}

\usepackage{stmaryrd}
\usepackage{amsthm}

\usetikzlibrary{patterns}

\usepackage{graphicx}
\usepackage{multicol}
\usepackage{breqn}
\usepackage{siunitx}
\usepackage[utf8]{inputenc}  
\usepackage[T1]{fontenc}     
\usepackage{microtype}
\microtypesetup{protrusion=false,expansion=true}
\usepackage[varqu]{zi4}

\usepackage{stmaryrd}

\usepackage{algorithm}

\usepackage{proof}

\newcommand{\highlight}[1]{\textcolor{black}{#1}}
\newtheorem{lemma}{Lemma}[section]

\tcbuselibrary{listings, skins}
\newtcblisting{scenicbox}[1]{
	colback=gray!5,      
	colframe=gray!40,    
	listing only,
	listing options={
		basicstyle=\ttfamily\footnotesize, 
		keywordstyle=\color{blue!80!black}\bfseries,
		commentstyle=\color{gray}\itshape,
		stringstyle=\color{green!60!black},
		numberstyle=\tiny\color{gray},
		breaklines=true,
		showstringspaces=false,
		tabsize=2,
		numbers=left,        
		numbersep=2pt,       
		xleftmargin=0pt,     
	},
	arc=2pt,             
	boxrule=0.5pt, 
	left=2pt, right=2pt, top=2pt, bottom=2pt, 
	enhanced, 
	#1
}

\newtheoremstyle{example_style} 
{3pt}  
{3pt} 
{\normalfont} 
{}     
{\itshape} 
{.}     
{ }    
{\thmname{#1}~\thmnumber{#2}} 

\theoremstyle{example_style} 
\floatname{algorithm}{Algorithm}

\makeatletter
\renewcommand{\ALG@beginalgorithmic}{\normalsize}
\algtext*{EndIf} 
\algtext*{EndFor} 
\algtext*{EndWhile} 
\makeatother

\newtheoremstyle{theory_style} 
{3pt}   
{3pt}   
{\normalfont} 
{}      
{\itshape} 
{.}    
{ }    
{\thmname{#1}~\thmnumber{#2}} 

\theoremstyle{theory_style}

\AtBeginDocument{%
	}

\setcopyright{cc}
\setcctype{by-nc-nd}
\acmDOI{10.1145/3839527}
\acmYear{2026}
\acmJournal{PACMPL}
\acmVolume{10}
\acmNumber{OOPSLA2}
\acmArticle{395}
\acmMonth{10}
\acmSubmissionID{oopslab26main-p1287-p}
\received{2026-03-17}
\received[accepted]{2026-08-06}
\begin{document}
	
		\title{Real-to-Sim Generation: Synthesizing Scenario Programs from Real-World Data via Constraint Solving}
	\author{Peishan Huang}
	\orcid{0009-0003-9800-3700}
	\affiliation{%
		\department{College of Computer Science and Technology}
		\institution{National University of Defense Technology}
		\city{Changsha}
		\state{Hunan}
		\country{China}
	}
	\email{huang\_ps@nudt.edu.cn}
	\authornote{Also with the affiliation: State Key Laboratory of Complex \& Critical Software Environment, National University of Defense Technology, Changsha, China.}
	
	\author{Wenmeng Zhang}
	\orcid{0009-0009-0646-6219}
	\affiliation{%
		\department{College of Computer Science and Technology}
		\institution{National University of Defense Technology}
		\city{Changsha}
		\state{Hunan}
		\country{China}
	}
	\email{wenmengzhang@nudt.edu.cn}
	\authornotemark[1]
	
	\author{Yusen Chen}
	\orcid{0009-0005-8839-4669}
	\affiliation{%
		\department{College of Computer Science and Technology}
		\institution{National University of Defense Technology}
		\city{Changsha}
		\state{Hunan}
		\country{China}
	}
	\email{chenyushen21@nudt.edu.cn}
	\authornotemark[1]

	\author{Zhenbang Chen}
	\orcid{0000-0002-4066-7892}
	\affiliation{%
		\department{College of Computer Science and Technology}
		\institution{National University of Defense Technology}
		\city{Changsha}
		\state{Hunan}
		\country{China}
	}
	\email{zbchen@nudt.edu.cn}
	\authornotemark[1]
	\authornote{Zhenbang Chen is the corresponding author.}
	
	\renewcommand{\shortauthors}{Peishan Huang, Wenmeng Zhang, Yusen Chen, Zhenbang Chen}
	\newcommand{\methodname}{\textsc{R2SGen}}
	\newcommand{\succRate}{\textsc{96.3\% }}
	
	\begin{abstract}
		The demand for synthetic training data is hindered by the sim-to-real gap, as current data-driven and LLM-based generators often produce physically implausible scenarios. To address this, we propose R2SGEN, a Real-to-Sim framework that synthesizes structured scenario programs from real-world data. To overcome the combinatorial explosion and intractability of monolithic Satisfiability Modulo Theories (SMT) encoding, we introduce a decoupled synthesis strategy. This approach separates the discrete structural program search from continuous geometric resolution using lightweight, atomic SMT constraints. Furthermore, we significantly accelerate the search process by integrating two tailored pruning mechanisms: Common Prefix Abstraction-based pruning for Breadth-First Search and Branch-and-Bound for Depth-First Search. We evaluate R2SGEN on 20 real-world scenes of varying complexity from the nuScenes dataset. 
		\highlight{Experimental results show that our method guarantees
			consistency with the input scene and produces substantially lower-cost
			programs than the LLM-based baselines under the evaluated inputs.}
		Both proposed search paradigms exhibit complementary advantages, proving highly efficient and scalable for high-complexity synthetic data generation.
	\end{abstract}


	\maketitle
	 

\section{Introduction}
Modern autonomous-driving systems require large volumes of visual data for 
training, yet collecting such data in the real world is costly and difficult to 
scale \cite{nuscenes,Waymo}. Synthetic data therefore provides a practical 
alternative for expanding training datasets \cite{Carla,lgsvl}.
Currently, image data generation primarily relies on neural network-based 
methods. However, these approaches exhibit significant limitations: their 
"black-box" nature lacks interpretability during the generation process, making 
it difficult to achieve fine-grained, customized control over specific semantic 
elements within a scene. In contrast, Scenic \cite{Scenic}, a domain-specific 
language (DSL) tailored for scenario generation, demonstrates remarkable 
advantages in constructing complex driving scenarios due to its programmability 
and determinism.

Despite its powerful capabilities, Scenic's nature as a DSL requires developers to possess deep domain expertise and manually write code, a process that is not only inefficient but also incurs high labor costs. 
Recently, some studies \cite{scenicNL,ChatScene} have attempted to leverage Large Language Models (LLMs) to automatically synthesize Scenic programs via natural language prompts. 
Nevertheless, constrained by the inherent hallucination issues of LLMs and the ambiguity of natural language descriptions, existing methods struggle to guarantee the correctness of the generated programs, let alone formally verify whether the synthesized Scenic programs strictly conform to the task specifications.
Consequently, there is an urgent need to develop a novel technical framework. This framework should automatically synthesize scenario programs that are consistent with real-world scene data. Furthermore, it must provide rigorous formal verification capabilities to ensure that the semantics of the finally generated code conform to the task specifications.

Our key insight is to formulate scenario generation as a Programming by Example problem. Given a concrete scenario from real-world data, comprising precise positions, headings, and types of all entities, our goal is to automatically synthesize a scenario program. This synthesized program must satisfy two requirements: first, it should describe the input scenario, meaning the program's semantics must encompass the given concrete configuration; second, it should generalize beyond the single input example, enabling the generation of diverse scenario variants.
We therefore propose \methodname\ (Real-to-Sim Generation), an automated
framework for scenario program synthesis from concrete scenes shown in Fig.
\ref{fig:OVERVIEW}. Our framework takes as input \highlight{a structured
concrete scene, \textit{i.e.}, a set of entities annotated with their types,
positions, and headings, as commonly provided by public datasets
(\textit{e.g.}, nuScenes) or by upstream perception systems or vision-language
models (VLMs)}, and produces as output a structured scenario program in Scenic 
that abstracts the concrete coordinates into reusable spatial relations and 
parameterized ranges, which can then be rendered into diverse synthetic images 
for training autonomous driving systems. However, realizing this synthesis 
poses two key challenges.

\begin{figure}[t]
	\centering
	\begin{minipage}[t]{0.95\textwidth}
		\centering
		\includegraphics[width=\textwidth]{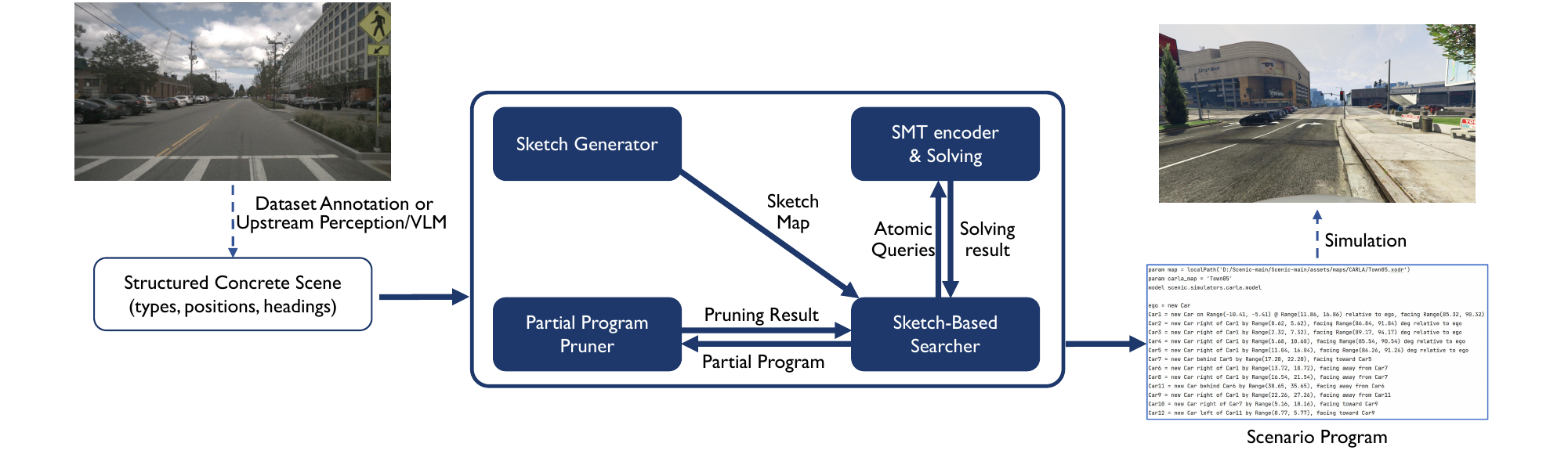}
		\vspace{-25pt}
		\caption{\highlight{Overview of R2SGen. Scene images are from the nuScenes dataset~\cite{nuscenes}, \copyright~Motional AD Inc., licensed under CC BY-NC-SA 4.0.}}
		\label{fig:OVERVIEW}
	\end{minipage}
\end{figure}

The first challenge is verifying program correctness. The DSL describes scenarios using spatial relationships such as "ahead of", but the input is given as exact numeric coordinates. Determining whether a candidate program correctly captures the input requires automatically reasoning about whether these relational descriptions are satisfied by the concrete coordinates.
For instance, given the traffic snapshot in Fig. \ref{fig:vis_input_a}, one could trivially generate programs using absolute coordinates or numerical offsets (top and middle of Fig. \ref{fig:prog_combined_b}).
 However, these approaches simply encode the specific numerical values without 
 capturing the underlying spatial relationships, such as "car3 is ahead of 
 ego". Such rigid encodings are scene-specific and cannot generalize: a program 
 encoding "car3 at (0, 5)" works only for that exact coordinate, whereas a 
 program encoding "car3 ahead of ego by Range(3,8)" captures a reusable spatial 
 pattern that can generate variations. Our goal is to synthesize programs like 
 the bottom example in Fig. \ref{fig:prog_combined_b}, which abstracts the 
 scene into relational primitives (\textit{e.g.}, \textit{ahead of} and 
 \textit{right of}) and parameterized ranges, enabling the program to generate 
 structurally similar scenarios beyond the single observed instance.
We draw inspiration from prior work \cite{Querying}, which encodes 
scenario programs into SMT constraints to validate their semantic alignment 
with real-world data. While their approach focuses on querying data instances 
that match a specific pre-defined program, we invert this paradigm to 
synthesize a scenario program directly from real-world data. We encode the 
geometric semantics of the DSL as satisfiability constraints and leverage SMT 
solvers to verify program consistency during the synthesis process. Since our 
synthesized programs are grounded in reality by design, they are inherently 
realistic, effectively mitigating the sim-to-real gap without requiring 
post-hoc validation.  

\begin{figure}[!t]
	\centering	
	\begin{subfigure}[b]{0.46\textwidth}
		\centering
		\includegraphics[width=\textwidth]{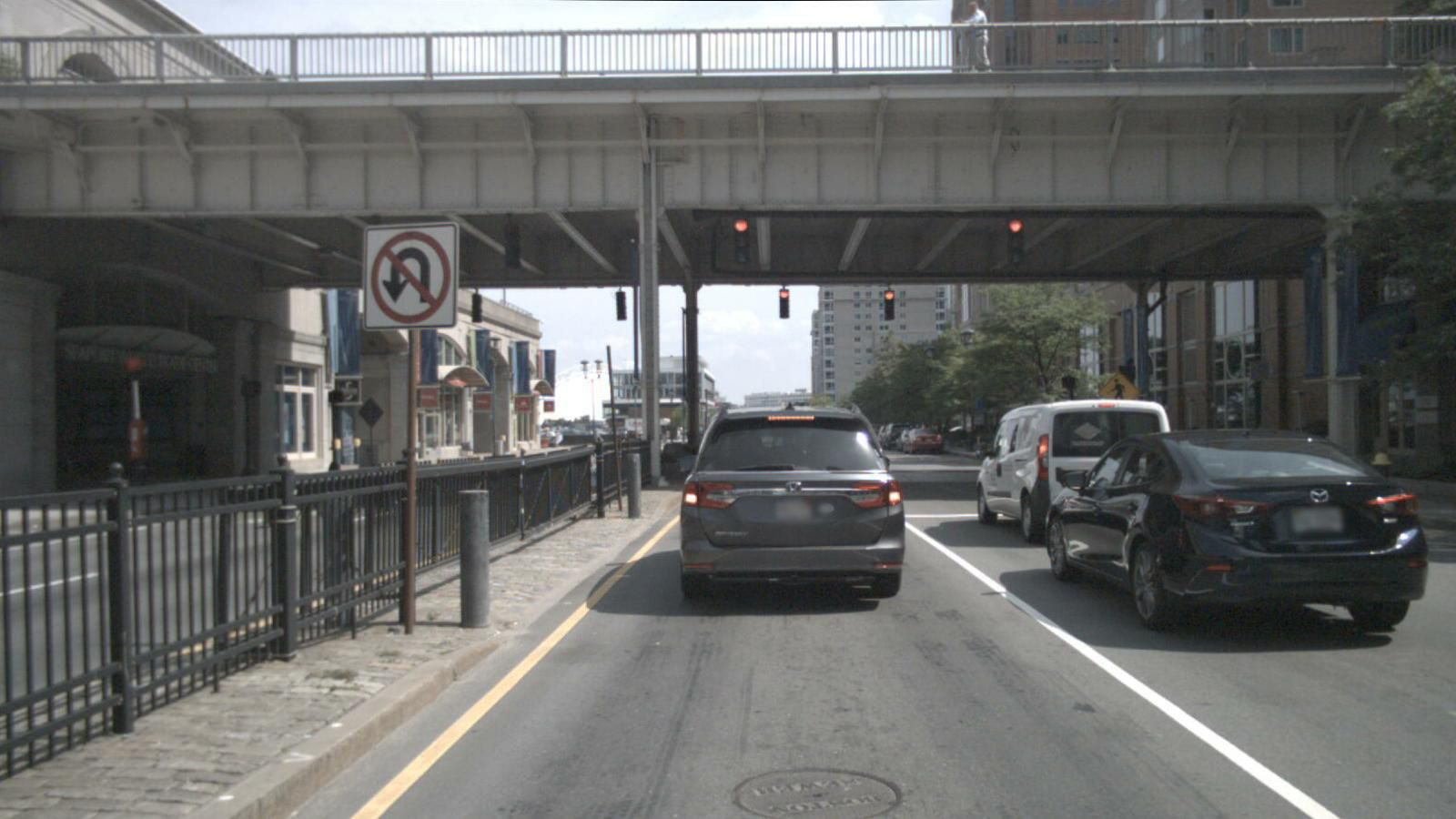}
		
		\vspace{10pt} 
		\begin{flushleft}
			\small
			\textbf{Hypothetical semantic information:} \\
			\vspace{4pt}
			\begin{tabular}{@{}ll@{}}
				\quad \ ego: $(0, 0, 0^\circ)$ &\quad car2: $(3, 5, 0^\circ)$ \\
				\quad \ car1: $(3, 8, 0^\circ)$ &\quad car3: $(0, 5, 0^\circ)$
			\end{tabular}
		\end{flushleft}
		
		\vspace{12pt} 
		\vspace{-3pt}
		\caption{Real-world scene image from the nuScenes 
			dataset~\cite{nuscenes} (\copyright~Motional AD Inc., licensed 
			under CC 
			BY-NC-SA 4.0) and the assumed corresponding semantic information.}
		\label{fig:vis_input_a}
	\end{subfigure}
	\hfill
	\begin{subfigure}[b]{0.50\textwidth}
\begin{lstlisting}[
		basicstyle=\ttfamily\fontsize{7.5}{9}\selectfont,
		frame=tb,
		xleftmargin=0pt,
		breaklines=false,
		commentstyle=\itshape\color{darkgray},
		escapechar=|,
		keepspaces=true,
		tabsize=4
		]
## Program with absolute coordinates
car1= new Car on |\textcolor{red}{3@8}|, facing |\textcolor{red}{0 deg}|
car2= new Car on |\textcolor{red}{3@5}|, facing |\textcolor{red}{0 deg}|
car3= new Car on |\textcolor{red}{0@5}|, facing |\textcolor{red}{0 deg}|

## Program with relative numerical offsets
car1= new Car on |\textcolor{orange}{3@8 relative to}| ego,
	facing |\textcolor{orange}{0 deg relative to}| ego
car2= new Car on |\textcolor{orange}{0@-3 relative to}| car1,
	facing |\textcolor{orange}{0 deg relative to}| car1
car3= new Car on |\textcolor{orange}{-3@0 relative to}| car2,
	facing |\textcolor{orange}{0 deg relative to}| car2

## Program capturing the scene layout
car3= new Car |\textbf{\textcolor{blue}{ahead of}}| ego by |\textbf{\textcolor{teal}{Range(3,8)}}|,
	facing |\textbf{\textcolor{blue}{away from}}| ego
car2= new Car |\textbf{\textcolor{blue}{right of}}| car3 by |\textbf{\textcolor{teal}{Range(1,6)}}|,
	facing |\textbf{\textcolor{teal}{Range(-5,5)}}| deg relative to ego
car1= new Car |\textbf{\textcolor{blue}{ahead of}}| car2 by |\textbf{\textcolor{teal}{Range(1,6)}}|,
	facing |\textbf{\textcolor{blue}{away from}}| car2
\end{lstlisting}
		\vspace{-3pt}
		\caption{Valid scenario programs describing the given 
			\highlight{structured scene}.}
		\label{fig:prog_combined_b}
	\end{subfigure}
	
	\vspace{-8pt}
	\caption{Comparison of spatial representations for a given scene. 
		(\textbf{a}) A real-world scene image from the ego vehicle's 
		perspective 
		(top) and the assumed corresponding semantic information (bottom). 
		(\textbf{b}) Three variations of valid scenario programs describing the 
		given scene: progressing from rigid absolute coordinates and relative 
		numerical offsets, to the optimal target program that captures the 
		inherent 
		relational structure.}
	\label{fig:synthesis_comparison}
\end{figure}

The second challenge is the vast combinatorial search space. Synthesizing scenario programs that capture relational structures faces a severe combinatorial explosion. A fundamental constraint is the dependency order: the position and orientation of any entity must be defined relative to the global frame or an already declared entity. Therefore, beyond merely solving for specific geometric relationships (\textit{e.g.}, $ahead\ of$ or $facing\ toward$), we must also determine which entity serves as the reference object. Consequently, based solely on the permutations of entity declarations, the search space complexity scales factorially ($O(N!)$).
Finding the optimal program in this vast space becomes computationally intractable as the number of entities increases. While one might intuitively encode the entire problem into a single, monolithic SMT formula, our empirical results show that the complexity of this joint discrete-continuous space causes solvers to fail even on simple three-entity scenes.
To overcome this, we decouple the synthesis process, fundamentally separating 
its discrete and continuous dimensions. We delegate discrete structural 
decisions, such as declaration ordering and relational operators, to an 
explicit search algorithm, while utilizing the SMT solver strictly as a 
geometric oracle for continuous constraints. By incrementally constructing the 
program via lightweight, atomic SMT queries, this decoupled approach makes the 
problem navigable, although the combinatorial space remains large.
To mitigate this combinatorial explosion, we introduce two tailored pruning strategies. First, \textbf{Common Prefix Abstraction-based Pruning}, primarily for \textbf{Breadth-First Search (BFS)}, merges equivalent states from distinct declaration orders to eliminate redundant branches. Second, \textbf{Branch-and-Bound Pruning}, exclusively for \textbf{Depth-First Search (DFS)}, \highlight{discards a partial branch when its accumulated cost is no lower than the cost of the best complete program found so far.}

We have implemented \methodname\ in a prototype tool and evaluated it on 20 real-world scenes of varying complexity sourced from the nuScenes \cite{nuscenes} dataset. 
\highlight{Under the evaluated inputs,
	\methodname\ achieves complete syntax and semantic correctness and
	synthesizes substantially lower-cost programs under our cost function,
	whereas the semantically correct baseline outputs primarily encode
	absolute coordinates.}
Furthermore, the efficiency evaluation validates the absolute necessity of our decoupled architecture and pruning strategies. Monolithic SMT encodings and naive combinatorial searches are fundamentally intractable, failing even on minimal configurations. In contrast, our tailored BFS and DFS search paradigms efficiently navigate the combinatorial explosion, successfully scaling to highly complex traffic scenes with up to 24 entities.
In summary, the contributions of this paper are as follows.

\begin{itemize}[leftmargin=10pt]
	\item 
	We propose a scenario program synthesis technique that, given \highlight{a structured concrete scene}, automatically generates a structured scenario program. This program replaces rigid numerical coordinates with symbolic spatial relations, precisely capturing the underlying geometric layouts and relational constraints of the real world.
	\item
	We introduce a decoupled synthesis strategy to overcome intractability of monolithic SMT encoding. By separation, we resolve continuous geometric relations via lightweight, atomic SMT constraints, while significantly accelerating discrete program search using two tailored pruning mechanisms: \textbf{Common Prefix Abstraction} for BFS and \textbf{Branch-and-Bound} for DFS.
	\item
	We conduct extensive evaluations on 20 diverse real-world scenes sourced from the nuScenes dataset. The results demonstrate our approach's high computational efficiency and scalability, successfully generating valid scenario programs for highly complex environments with up to 24 entities where existing baselines fail.
\end{itemize}

	%

\vspace{-6pt}
\section{System Overview}

In this section, we present the overall architecture of \methodname, as illustrated in Fig. \ref{fig:OVERVIEW}, and detail the functionality of its key components.

\textbf{Structured Concrete Scene.} 
\highlight{To bridge the sim-to-real gap, \methodname\ starts from a structured concrete scene. As illustrated in Fig.~\ref{fig:vis_input_a}, the input describes a real-world traffic scene containing three entities, each annotated with its type, position, and heading. Such a structured scene can be obtained from annotated datasets or from upstream perception/VLM frontends. \methodname\ synthesizes a generalized program from this input, which can subsequently be executed by a simulator to generate diverse scenario variations that preserve the core spatial relations of the original scene.}

\textbf{Expected Output.} 
As an important scenario description language, Scenic is compatible with multiple simulators and can produce corresponding scene images through simulation. Therefore, \methodname\ targets the synthesis of Scenic program that captures the spatial layouts between entities in the \highlight{input scene}, enabling the generation of additional images that preserve these spatial relations via these simulators.
Our approach focuses on a core subset of Scenic. Rather than merely transcribing hard-coded numerical coordinates, our approach aims to discover the most natural relational dependencies (\textit{e.g.}, ahead of, facing toward) to describe the scene layout.
Given an \highlight{input scene}, our approach synthesis a Scenic program such that the \highlight{input scene} lies within the scene space defined by that program. 
As illustrated in Fig. \ref{fig:prog_combined_b}, while the scene in Fig. \ref{fig:vis_input_a} can be represented by various valid programs, \methodname\ is designed to search for and synthesize the most structured representation. Specifically, \methodname\ aims to output the bottom program as the optimal result, as it replaces fixed numerical values with symbolic spatial relations to capture the scene's core layout.

\textbf{Sketch Generator.} 
Because the synthesized target program is strictly scoped to the entities populating the scene, it structurally comprises exactly $n$ entity-specific declarations. Given $n$ distinct entities extracted from the input scene, we directly construct a structural program sketch containing $n$ uninstantiated statements (\textit{i.e.}, syntactic holes), as demonstrated below. 
\vspace{-2pt}
\begin{equation}
	\small
	\begin{aligned}
		\square_{stmt} \  \ \square_{stmt} \  \ 
		\dots \  \ \square_{stmt} 
	\end{aligned}
	\label{eq:sketch_example2}
\end{equation}

\noindent Each entity-specific declaration 
$\square_{stmt}$ is defined as 
$\square_{entity} = \textsf{new} \ 
\square_{type} \  \square_{pos} , \ 
\square_{heading} $. 
Specifically, $\square_{entity}$ denotes the target entity 
to be declared, and $\square_{type}$ denotes the type of 
the entity.
$\square_{pos}$ and $\square_{heading}$ represent the 
relative positional and orientational relationships, 
respectively.

\textbf{Sketch-Based Searcher.}
The sketch-based search process fundamentally amounts to exploring the program space defined by the sketch, iteratively filling in the hole nodes within it. We primarily employ two distinct search strategies: breadth-first search (BFS) and depth-first search (DFS). 
Because positions and orientations are expressed relative to a reference frame, the search process is designed to address the following three parts:
\begin{itemize}[leftmargin=10pt]
	\item \textbf{Position attribute resolution}: The searcher 
	prioritizes more direct and natural positional relation predicates 
	(such as \textit{ahead of} and \textit{right of}) over explicit 
	coordinate values. Concurrently, it infers feasible numerical 
	ranges associated with these relations (\textit{e.g.}, “5 to 10 
	meters ahead”). Explicit coordinate attributes (\textit{i.e.}, 
	concrete Cartesian positions) are assigned lower priority and incur 
	higher synthesis cost compared to relational predicates.
	\item \textbf{Orientation attribute resolution}: Analogous to position, the searcher prioritizes relational orientations (such as \textit{facing toward} or \textit{facing away from}) over explicit angular offsets.
	\item \textbf{Selection of the relative reference frame}: When multiple entities are present, only the ego vehicle or declared entities may serve as valid reference frames, ensuring dependency chain validity. The resolution of the optimal positional and orientational attributes is intrinsically dependent on the choice of the reference entity. Specifically, we evaluate candidate references and select the one that yields the highest-priority spatial relation.
\end{itemize}

\textbf{SMT Encoder \& Solving.}
Building upon the exploration of the searcher, the exact instantiation of statement attributes is delegated to an SMT solver. Specifically, for a given target and reference entity, the synthesizer iterates through candidate spatial relations, systematically encoding their geometric semantics into SMT formulas. The solver acts strictly as a geometric oracle to verify the satisfiability of these hypothesized relations. Upon returning a \texttt{sat} determination, the solver extracts a valid model, directly yielding the concrete numerical assignments for the continuous parameters (\textit{e.g.}, distance boundaries or angular ranges).

\textbf{Partial Program Pruner.}
As the number of entities increases, the synthesis search space expands factorially ($O(N!)$) due to the permutation of insertion orders. As the complexity of the target scene increases, the search space (despite being guided by a sketch) remains prohibitively large. To accelerate the search process, we introduce search pruning techniques. Specifically, we integrate the following two complementary pruning mechanisms:
\begin{itemize}[leftmargin=10pt]
	\item \textbf{Common Prefix Abstraction-based Pruning.} 
	When declaring a new entity, its description must be anchored to either the ego vehicle or an already-declared entity as the reference. Given the same set of declared entities, selecting the same reference for a new entity incurs identical incremental cost.	
	Consider the example in Fig. \ref{fig:vis_input_a}, where the input scene contains three vehicles (denoted as car1, car2, and car3). Suppose we have two partial programs: 
	$ P_a =\{\text{car1}\ \ \text{car2}\ \ \square_{stmt} 
	\} $ where car1 and car2 have been declared in that 
	order , and 
	$ P_b =\{\text{car2}\ \ \text{car1}\ \ \square_{stmt} 
	\} $ which declares the same two entities but in a 
	different order.	
	Both $P_a$ and $P_b$ have declared the same set of entities (car1 and car2), so any subsequent entity (\textit{e.g.}, car3) will be described relative to the same reference entity under our semantics, incurring identical incremental cost. If the total cost of $P_a$ is lower than that of $P_b$, then for any valid completion of the hole, the full program derived from $P_a$ will always have a lower (or equal) total cost than the corresponding completion of $P_b$. Therefore, $ P_b$ can be safely pruned from the search space without sacrificing optimality.
	\item \textbf{Branch-and-Bound Pruning.} 
	Branch-and-bound pruning is an exact algorithmic strategy commonly used for solving combinatorial optimization problems. Its core idea is to systematically enumerate candidate solutions while leveraging upper- and lower-bound information to eliminate subspaces that cannot contain an optimal solution, thereby avoiding exhaustive enumeration.	
	In our setting, we maintain the best complete candidate program $P$ found so far (\textit{i.e.}, the one with the lowest total cost). During search, if a partial program’s accumulated cost meets or exceeds the cost of $P$, then any full program derived by extending this partial program will necessarily have a cost no less than that of $P$. Consequently, such a partial program can be safely pruned from the search space without risking the loss of an optimal solution.
\end{itemize}
	%

\section{Preliminaries}

This chapter establishes the formal foundations of our scenario synthesis 
framework. We first introduce the syntax and semantics of our DSL. Building 
upon these definitions, we provide a rigorous formulation of the scenario 
synthesis problem. 

\vspace{-3pt}
\subsection{DSL Syntax}
\vspace{-3pt}
We formally define the syntax of our DSL in Figure~\ref{fig:dsl_syntax}. 
Derived from the Scenic, the DSL is designed to model the spatial layout of 
entities by defining the relative relationships between them.
\highlight{Our DSL deliberately focuses on static spatial layouts,
	retaining Scenic primitives for describing relative positions and
	headings among entities. For numerical variation, it uses
	\textbf{Range} as its sole probabilistic operator, representing uniform
	distributions over distance and heading intervals. Dynamic behaviors,
	temporal interactions, and the richer probabilistic and procedural
	constructs of full Scenic are outside the current scope.}
A program $\langle P \rangle$ consists of a sequence of entity statements $\langle S \rangle$. Each statement $\langle S \rangle$ instantiates an entity and defines its state via three core components:
\begin{itemize}[leftmargin=10pt]
	\vspace{-1pt}
	\item \textbf{Type} ($\langle Type \rangle$): Specifies the category of the instantiated object (\textit{e.g.}, \textbf{Car}, \textbf{Pedestrian}, ...). 
	\item \textbf{Position} ($\langle Pos \rangle$): 
	Specifies the spatial location of an entity using three distinct modes of representation. First, we favor high-level directional primitives (\textit{e.g.}, \textbf{ahead of}, \textbf{left of}) that establish semantic spatial relations relative to a reference entity. These semantic primitives are assigned higher priority and incur lower synthesis cost. Second, the syntax supports explicit local Cartesian coordinates defined directly within a reference entity's local frame (\textbf{on} \dots \textbf{relative to}). Third, it allows for absolute global Cartesian coordinates (\textbf{on} \dots @ \dots) to place entities independently in the global space.
	\item \textbf{Heading} ($\langle Heading \rangle$): 
	Specifies an entity’s orientation via three distinct modes. First, we favor high-level orientational primitives (\textit{e.g.}, \textbf{facing toward}, \textbf{facing away from}) establishing semantic alignment relative to a reference entity. These primitives are assigned higher priority and incur lower synthesis cost. Second, the syntax supports explicit angular ranges defined relative to a reference entity's heading (\textbf{facing} \dots \textbf{deg relative to}). Third, it allows absolute global angular ranges (\textbf{facing} \dots \textbf{deg}) to orient entities independently in the global coordinate system.
\end{itemize}
\vspace{-3pt}

\begin{figure}[t]
	\centering
	\small 
	\begin{tabular}{rcl}
		$\langle P \rangle$ & $::=$ & $ \langle S \rangle^*$ \\[3pt]
		
		$\langle S \rangle$ & $::=$ & $\langle Entity \rangle = \textbf{new} \ \langle Type \rangle \ \langle Pos \rangle, \ \langle Heading \rangle$ \\[3pt]
		
		$\langle Type \rangle$ & $::=$ & $\textbf{Car} \mid \textbf{Pedestrian} \mid \dots$ \\[3pt]
		
		$\langle Pos \rangle$ & $::=$ & $(\textbf{ahead of} \mid \textbf{behind} \mid \textbf{left of} \mid \textbf{right of}) \ \ \langle Ref \rangle \ \ \textbf{by } \langle Range \rangle$ \\
		&\ \ &$\mid$  $\textbf{on } \langle Range \rangle @ \langle Range \rangle \ \textbf{relative to } \langle Ref \rangle $ $\mid$  $\textbf{on } \langle Range \rangle @ \langle Range \rangle$\\[3pt]
		
		$\langle Heading \rangle$ & $::=$ & $\textbf{facing } (\textbf{toward} \mid \textbf{away from}) \ \langle Ref \rangle$  $\mid$  $\textbf{facing } \langle Range \rangle \textbf{ deg relative to } \langle Ref \rangle$ \\ & \ &$\mid$  $\textbf{facing } \langle Range \rangle \textbf{ deg }$\\[5pt]

		$\langle Range \rangle$ & $::=$ & $\textbf{Range}(a, b), \quad a, b \in \mathcal{R}$ \\[3pt]
		
		$\langle Ref \rangle$ & $::=$ & $\langle Entity \rangle \mid \textbf{ego}$ \\[3pt]
		
		$\langle Entity \rangle$ & $::=$ & ${id_1} \mid {id_2} \mid \dots$ 
	\end{tabular}
	\vspace{-10pt}
	\caption{The syntax of the Scenic DSL, where $\mathcal{R}$ denotes the set of 
		real numbers. A program consists of $n$ entities, each specified by a tuple of 
		three core attributes: its type, spatial position, and orientation.}
	\label{fig:dsl_syntax}
\end{figure}

To enable generalization over the \highlight{structured concrete scene}, we employ the parameter $\langle Range \rangle$ to construct interval-based uniform distributions within the real-valued domain $\mathcal{R}$, thereby describing a continuous space of plausible scenes. Although the syntax permits referencing arbitrary identifiers, we impose a strict semantic constraint: an entity may only use the \textbf{ego} (\textit{i.e.}, the self-entity) or previously declared entities as its reference frame.

\vspace{-2pt}
\subsection{DSL Semantics and SMT Encoding}
\vspace{-2pt}

\subsubsection{Basic Definition}
Before presenting the specific denotational semantics, we formalize the fundamental notations. 
\vspace{-4pt}
\begin{definition}[Entity State Space]
	\label{def:state_space}
	The state space of an entity is defined as $\mathcal{X} = \mathcal{T} \times \mathcal{R}^2 \times \Theta$, where $\mathcal{T}$ is a finite set of entity types ($\mathcal{T}$ contains types such as Car and Pedestrian), and $\Theta = [0, 360^\circ)$ denotes the heading angle. We establish a 2D Cartesian coordinate system where the heading angle $\theta$ is measured counter-clockwise starting from the positive Y-axis (North). A concrete state $x = (t, p, \theta) \in \mathcal{X}$ comprises a type $t \in \mathcal{T}$, a position $p \in \mathcal{R}^2$, and a heading $\theta \in \Theta$.
\end{definition}
\vspace{-6pt}
\begin{definition}[Scene]
	\label{def:scene}
	Let $\mathcal{E}$ be a finite set of available entities. A concrete scene $\sigma$ is a partial function assigning each involved entity a state, formally defined as $\sigma: \mathcal{E} \rightharpoonup \mathcal{X}$. We denote the universe of all such scenes as $\Sigma$. The scene containing no entities is denoted the empty scene $\sigma_{\emptyset} \in \Sigma$.
\end{definition}
\vspace{-4pt}
\begin{definition}[Scenario]
	\label{def:scenario}
	A scenario $\Omega$ is formalized as a set of concrete scenes, \textit{i.e.}, $\Omega \subseteq \Sigma$. The semantic domain of scenarios is defined as the powerset $2^{\Sigma}$. We specifically define the empty scenario as $\Omega_{\emptyset} = \{ \sigma_{\emptyset} \} \in 2^{\Sigma}$.
\end{definition}
\vspace{-4pt}

\subsubsection{Denotational Semantics}
The denotational semantics maps the syntactic program into a scenario. The semantic functions are defined as follows:
\vspace{-2pt}
\begin{align*}
	\mathbb{T} &: \langle Type \rangle \to \mathcal{T} & \mathbb{H} &: \langle Heading \rangle \times \mathcal{X} \times \mathcal{R}^2 \to 2^\Theta \\[-2pt]
	\mathbb{R} &: \langle Range \rangle \to 2^{\mathcal{R}} & \mathbb{V} &: \langle S \rangle \times 2^\Sigma \to 2^\Sigma \\[-2pt]
	\mathbb{P} &: \langle Pos \rangle \times \mathcal{X} \to 2^{\mathcal{R}^2} & \mathbb{D} &: \langle P \rangle \times 2^\Sigma \to 2^\Sigma
\end{align*}

The base valuations $\mathbb{T}$ and $\mathbb{R}$ deterministically map syntactic primitives to their fundamental mathematical domains: $\mathbb{T}$ maps a type identifier to its corresponding semantic element in $\mathcal{T}$, and $\mathbb{R}(\textbf{Range}(a,b)) = \{r \in \mathcal{R} \mid a \le r \le b\}$, where $a, b \in \mathcal{R}$ are real numbers.

Crucially, the statement semantics $\mathbb{V}$ incrementally constructs the scenario. Let $S$ be an entity declaration statement of the form $e = \textbf{new}\ type\ pos, heading$. Given a current scenario $\Omega$, $\mathbb{V}$ evaluates $S$ by extending every existing scene $\sigma \in \Omega$ with the newly instantiated entity $e$:
{\small
$$
\mathbb{V}(S, \Omega) = \left\{ \sigma \cup \{e \mapsto (t, p, \theta)\}  \middle|\ \sigma \in \Omega \land t = \mathbb{T}(type) \land p \in \mathbb{P}(pos, s_{pos})  \land \theta \in \mathbb{H}(heading, s_{heading}, p) \right\}
$$}

\noindent where $s_{pos} = \sigma(ref_{pos})$ and $s_{heading} = \sigma(ref_{heading})$ are the states of the reference entities specified in $pos$ and $heading$ respectively (or undefined for absolute coordinates).

The program semantics $\mathbb{D}$ is then the composition of individual statement valuations. For a program $P = S \ P'$, the scenario is progressively constructed:$$\mathbb{D}(S \ P', \Omega) = \mathbb{D}(P', \mathbb{V}(S, \Omega))$$Ultimately, the final scenario is derived by evaluating the program starting from the initial empty scenario: $\Omega_{P} = \mathbb{D}(P, \Omega_{\emptyset})$.

\textbf{Semantics of Position ($\mathbb{P}$)}
We formalize the positional primitives (spanning high-level directional relations, explicit relative Cartesian coordinates, and absolute global placements) by defining their denotational semantics. The semantic function $\mathbb{P}$ resolves these spatial specifications into concrete geometric regions conditionally evaluated given the state of the reference entity. For relative primitives, let $s_{e} = (t_e, (x_e, y_e), \theta_e) \in \mathcal{X}$ denote the state of the reference entity $e$.
High-level directional relations permit a tolerance in angular alignment. We introduce a tolerance parameter $\delta_1$ to define the acceptable angular deviation. In our implementation, we empirically set $\delta_1$ to 10°. This bounds the valid spatial region to a $\pm$10° fan-shaped sector strictly aligned with the reference entity's forward directional axis. For example, consider the semantic evaluation of the expression "\textbf{ahead of} $e$". This expression restricts the valid target coordinates to a narrow $20^\circ$ cone directly in front of the reference entity $e$.
As shown in \eqref{eq:pos_semantics}, the semantic evaluation of a specific directional expression ($dir \in \{\textbf{ahead of}, \textbf{behind}, \textbf{left of}, \textbf{right of}\}$) denotes the subset of coordinates in $\mathbb{R}^2$ that satisfy the distance and angular constraints relative to the exact state of $e$.
\begin{equation}
	\small
	\label{eq:pos_semantics}
\mathbb{P}(dir \ e \textbf{ by } \textbf{Range}(a,b), s_{e}) =
\left\{ (x, y)  \ \middle|\
\begin{aligned}
	& \sqrt{(x-x_e)^2 + (y-y_e)^2} \in \mathbb{R}(\textbf{Range}(a,b)) \\
	\land \ & \mathcal{W}\left( \mathcal{N}(\beta((x,y), (x_e, y_e))), \ \mathcal{N}(\theta_e + \phi_{dir}) \right) \le \delta_1
\end{aligned}
\right\}
\end{equation}
Where:
\begin{itemize}[leftmargin=10pt]
	\item $\beta((x,y), (x_e,y_e))$ computes the absolute angle from the reference point $(x_e,y_e)$ to the target point $(x,y)$. Since the standard $atan2$ function (assumed to output degrees) measures the angle relative to the positive X-axis, we apply a $-90^\circ$ offset to align the result with our coordinate system, where $0^\circ$ is defined as the positive Y-axis (North):
	\[\small
	\beta((x,y), (x_e,y_e)) = atan2(y - y_e, x - x_e) - 90^\circ
	\]
	
	\item $\mathcal{N}(\alpha)$ is the normalization function that maps any angle $\alpha$ to the range $[0, 360)$
	\[\small
	\mathcal{N}(\alpha) = 
	\begin{cases} 
		\alpha - 360^\circ & \text{if } \alpha \ge 360^\circ \\
		\alpha + 360^\circ & \text{if } \alpha < 0^\circ \\
		\alpha & \text{otherwise}
	\end{cases}
	\]
	
	\item $\mathcal{W}(u, v)$ normalizes the absolute difference between two angles $u, v \in [0, 360)$ to the range $[0, 180^\circ]$.
	\[\small
	\mathcal{W}(u, v) = 
	\begin{cases} 
		360^\circ - |u - v| & \text{if } |u - v| > 180^\circ \\
		|u - v| & \text{otherwise}
	\end{cases}
	\]
	
	\item $\phi_{dir}$ specifies the base angular offset for the target direction. These values represent the relative rotation applied to the reference entity's current heading $\theta_e$, strictly following our counter-clockwise coordinate convention:
	\[
	\phi_{ahead} = 0^\circ, \quad \phi_{left} = 90^\circ, \quad \phi_{behind} = 
	180^\circ, \quad \phi_{right} = 270^\circ 
	\]
\end{itemize}

The evaluation of the relative expression $\textbf{on Range}(a,b) @ \textbf{Range}(c,d) \textbf{ relative to } e$ denotes the subset of coordinates in $\mathcal{R}^2$ that fall within the Cartesian product of longitudinal and lateral intervals defined in the reference entity's local frame:
\begin{equation}
	\small
	\label{eq:local_xy_semantics}
\begin{split}
	\mathbb{P}(\textbf{on } \textbf{Range}(a,b)@&\textbf{Range}(c,d) \textbf{ relative to } e, s_{e}) = \\
	&\quad \left\{ (x, y)  \ \middle|\
	\begin{aligned}
		& ((x - x_e) \cos \theta_e + (y - y_e) \sin \theta_e) \in \mathbb{R}(\textbf{Range}(a,b)) \\
		\land \ & ((y - y_e) \cos \theta_e - (x - x_e) \sin \theta_e) \in \mathbb{R}(\textbf{Range}(c,d))
	\end{aligned}
	\right\}
\end{split}
\end{equation}

The absolute Cartesian primitive defines a spatial bounding box directly in the global coordinate system, operating independently of any reference entity:
\begin{equation*}
	\small
	\mathbb{P}(\textbf{on } \textbf{Range}(a,b) @ \textbf{Range}(c,d)) =
	\left\{ (x, y) \ \middle|\ x \in \mathbb{R}(\textbf{Range}(a,b)) \land y \in \mathbb{R}(\textbf{Range}(c,d)) \right\}
\end{equation*}

\textbf{Semantics of Heading ($\mathbb{H}$)}
We formalize the heading primitives (spanning high-level interactive orientations, explicit relative angles, and absolute global angles) by defining their denotational semantics. The valuation function $\mathbb{H}$ denotes the set of valid absolute angles in $[0, 360^\circ)$ that satisfy the given constraints.

To similarly capture the semantic looseness of orientational constraints like \textbf{facing toward} or \textbf{facing away from}, we introduce a heading tolerance parameter $\delta_2$, which relaxes strict geometric alignment into a valid continuous interval. In our implementation, we empirically set $\delta_2$ to $5^\circ$.

Crucially, computing these interactive orientations inherently depends on both the exact state of the reference entity and the spatial location of the current entity itself. Let $\sigma \in \Sigma$ be the concrete scene currently being evaluated, which fixes the reference entity's state as $\sigma(e) = (t_e, (x_e, y_e), \theta_e)$. Furthermore, let $p = (x, y) \in \mathcal{R}^2$ represent a candidate position drawn from the entity's positional feasible region (as bound in the statement semantics $\mathbb{V}$). The denotational semantics evaluates the valid angular intervals conditional on both the scene $\sigma$ and the candidate position $p = (x, y)$:
\vspace{-3pt}
\begin{equation}
	\small
	\label{eq:heading_semantics}
	\mathbb{H}(heading, s_{e}, (x, y))=
	\begin{cases}
		\left\{ \phi \middle|\
		\highlight{\mathcal{W}\left(
			\phi,
			\mathcal{N}\left(\beta((x_e,y_e),(x,y))\right)
			\right) \le \delta_2}
		\right\} \\
		\quad  \text{if } heading =\textbf{facing toward } e \\[5pt]

		\left\{ \phi  \middle|\
		\highlight{\mathcal{W}\left(
			\phi,
			\mathcal{N}\left(\beta((x,y),(x_e,y_e))\right)
			\right) \le \delta_2}
		\right\} \\
		\quad   \text{if } heading =\textbf{facing away from } e
		\\[5pt]

		\left\{ \phi  \middle|\ \exists \alpha \in \mathbb{R}(\textbf{Range}(a,b)), \ \phi = \mathcal{N}(\theta_e + \alpha)  \right\} \\
		\quad \text{if } heading = \textbf{facing}\ \textbf{Range}(a,b) \textbf{ deg relative to } e \\[5pt]

		\left\{ \phi  \ \middle|\ \exists \alpha \in \mathbb{R}(\textbf{Range}(a,b)), \ \phi = \mathcal{N}(\alpha)  \right\} \\ \quad
	 \text{if } heading =\textbf{facing}\ \textbf{Range}(a,b) \textbf{ deg}
	\end{cases}
\end{equation}

\vspace{-3pt}

\textbf{SMT Encoding of Position ($\mathbb{P}$) and Heading ($\mathbb{H}$)}
Our SMT formulation computationally realizes these semantics, serving 
as a constraint-based oracle to verify the feasibility of candidate 
spatial and orientational configurations.

For position encoding, given a candidate coordinate $(x, y)$ and reference state $(x_e, y_e, \theta_e)$, the solver evaluates the satisfiability of the geometric membership. For high-level directional relations (\textbf{ahead of}, \textbf{behind}, \textbf{left of}, \textbf{right of}), the solver performs conditional spatial checks to verify if geometric membership holds within fuzzy angular regions. In contrast, explicit relative and absolute coordinates are mathematically guaranteed to be \texttt{SAT}; since the physical states are fixed during evaluation, their corresponding coordinates can always be determined. Upon a \texttt{SAT} result, we extract valid boundaries to instantiate spatial ranges.

For heading encoding, high-level heading relations (\textbf{facing toward}/\textbf{facing away from}) are evaluated by checking the satisfiability of the orientation variable $\phi$ bounded by the tolerance $\delta_2$. For numerical primitives (\textbf{facing Range}(a,b) \textbf{deg relative to} and \textbf{facing Range}(a,b) \textbf{deg}), the constraints explicitly bound the valid angular space either as an offset from the reference's heading or within the global coordinate system, from which the solver extracts the angular boundaries.

\highlight{Although an individual relation over fixed coordinates can
	be checked using closed-form arithmetic, we use SMT to express relation
	feasibility and the synthesis of distance or angle interval parameters
	in a uniform nonlinear-real constraint system. This formulation also
	provides a direct extension path for jointly reasoning about additional
	continuous constraints, such as uncertain input regions or map/lane
	constraints.}

To accelerate solving and ensure practical alignment, we employ unified fixed-length constraints: for positional ranges, $b = a + C_{dist}$ and $d = c + C_{dist}$ with $C_{dist}=5\,m$; for heading ranges, $b = a + C_{head}$ with $C_{head} = 5^\circ$.
\highlight{These four constants ($\delta_1$, $\delta_2$, $C_{dist}$, and 
$C_{head}$) serve two different roles. The angular
	tolerances $\delta_1$ and $\delta_2$ define the acceptable angular
	deviations for high-level directional and heading relations. Increasing
	these tolerances allows such relations to cover larger angular regions,
	whereas decreasing them requires closer geometric alignment and may
	cause the synthesizer to use explicit relative coordinates or angles
	instead. In contrast, $C_{dist}$ and $C_{head}$ determine the widths of
	the synthesized \textbf{Range} expressions: larger values allow broader
	variations in position or heading, while smaller values keep sampled
	scenes closer to the input. We use the same values throughout our
	evaluation; other domains or entity scales may use different values to
	reflect their spatial characteristics and the desired amount of
	scenario variation.}

\vspace{-5pt}
\subsection{Problem Definition}


We cast the generation of scenario programs as a program synthesis task, specifically searching for the optimal scenario program whose semantic domain contains the given concrete scene.

\begin{definition}
	\textbf{\textup{(Specification)}} We utilize the input scene $\sigma_{in} \in \Sigma$ as a formal specification for program synthesis, thereby guiding the search process toward valid program configurations.
\end{definition}

\begin{definition}
	\label{consistency}
	\textbf{\textup{(Consistency)}} We define a consistency relation between a scenario program $P$ and a specification $\sigma_{in}$ as $P \models \sigma_{in}$. This relation holds if and only if $\sigma_{in}$ resides within the semantic domain induced by $P$, \textit{i.e.}, $ \sigma_{in} \in \mathbb{D}(P, \Omega_{\emptyset})$.
\end{definition}

\begin{definition}
	\label{costF}
	\textbf{\textup{(Cost Function $\mathcal{C}$)}} To rank semantically equivalent programs by how well they capture scene relational structures, we define a cost function over the $n$ statements of $P$ as  \[\mathcal{C}(P) = \sum_{i=1}^{n} \big( penalty(pos_i) + penalty(heading_i) \big)\]
	\vspace{-3pt}
	The penalty prioritizes high-level relational semantics:
	\[
	penalty(pos_i) = 
	\begin{cases} 
		2 & \text{if } \mathit{pos}_i \text{ = } \textbf{on } \langle Range \rangle @ \langle Range \rangle  \\
		1 & \text{if } \mathit{pos}_i \text{ = } \textbf{on } \langle Range \rangle @ \langle Range \rangle \textbf{ relative to } \langle Ref \rangle \\
		0 & \text{if } \mathit{pos}_i \text{ = } dir \ \langle Ref \rangle \textbf{ by } \langle Range \rangle
	\end{cases}
	\]
	\[
	penalty(\mathit{heading}_i) = 
	\begin{cases} 
		2 & \text{if } \mathit{heading}_i \text{ = } \textbf{facing } \langle Range \rangle \\
		1 & \text{if } \mathit{heading}_i \text{ = } \textbf{facing } \langle Range \rangle \textbf{ deg relative to } \langle Ref \rangle \\
		0 & \text{if } \mathit{heading}_i \text{ = } \textbf{facing } (\textbf{toward} \mid \textbf{away\ from}) \ \langle Ref \rangle
	\end{cases}
	\]
\end{definition}

\vspace{-2mm}
\begin{definition}
	\label{Problem}
	\textbf{\textup{(Problem)}} The scenario program synthesis task is formalized as a combinatorial optimization problem. We seek to find an optimal valid program $\mathcal{P}^*$ that minimizes $\mathcal{C}(P)$, under the strict constraint that the denotational semantics of $\mathcal{P}^*$ contains the input concrete scene $\sigma_{in}$.
		\begin{equation}
			\mathcal{P}^* = \underset{P}{\operatorname{arg\ min}} \ \mathcal{C}(P) \quad \text{s.t.} \quad \mathcal{P}^* \models \sigma_{in}
		\end{equation}
	\vspace{-2mm}
\end{definition}

	%

\vspace{-4mm}
\section{Synthesis Algorithm}
\label{sec:synthesis_algorithm}

This section presents the algorithm for the scenario synthesis problem. We first introduce the top-level algorithm, followed by its other components.

\vspace{-2mm}
\subsection{Top-Level Synthesis Algorithm}
\label{subsec:top_level_algorithm}

	Alg. \ref{alg:top_level} outlines our top-level synthesis procedure, \textsc{ConstraintSolvingBasedSynthesizer}. 
It accepts two inputs: a concrete scene $\sigma_{in}$ (extracted from a real image, containing state labels for a finite set of entities $\mathcal{E}$) and a specified search strategy \text{mode} $\in \{\text{DFS, BFS}\}$. 
The objective is to output an optimal scenario program $\mathcal{P}^*$  rigorously satisfying given specifications (Definition \ref{Problem}).  
The algorithm begins by invoking \textsc{SketchGenerator} (Line \ref{sg}) to construct a structural sketch $\mathcal{S}$ from $\sigma_{in}$. 
Specifically, using $\mathcal{E}$’s cardinality, this generator instantiates an equal number of statements with uninstantiated holes, as previously illustrated in Eq.(\ref{eq:sketch_example2}). 
Following the sketch generation, the algorithm routes execution to either \textsc{SketchBasedSearchBFS} (Line \ref{bfs}) or \textsc{SketchBasedSearchDFS} (Line \ref{dfs}), depending on the specified \text{mode}, to systematically resolve the holes and discover $\mathcal{P}^*$. 
Crucially, assuming a perfect SMT oracle, our search framework guarantees soundness and completeness. Formal discussions and proofs of these properties are detailed in Section \ref{subsec:theoretical_analysis}.

\vspace{-2pt}
\subsection{Sketch-Based Search with Pruner} 
\label{subsec:sketch_based_search}

In this subsection, we formulate the scenario program 
generation process as a combinatorial search problem over 
the permutation space of entity declarations. As formalized 
in Eq.~\ref{eq:sketch_example2}, the scenario sketch 
$\mathcal{S}$ is conceptualized as an ordered sequence of 
uninstantiated statement holes ($\square_{stmt}$). The 
primary objective is to dynamically determine the optimal 
declaration order for the finite set of entities 
$\mathcal{E}$, while minimizing the cumulative semantic 
penalty $C$. To efficiently navigate this permutation 
space, we implement two variant search strategies: 
Breadth-First Search (BFS) and Depth-First Search (DFS).

Both search strategies operate on a unified state-space tree representing the permutation of different entities. In this tree, each node represents the instantiation of a single entity's statement. Consequently, the continuous path from the root to any given node constitutes a search state, which is formally defined by a partial sketch $\mathcal{S}_{curr}$ and its accumulated semantic penalty $C_{curr}$.

At each expansion step, the search controller determines 
the set of defined entities $E_{curr}$ within 
$\mathcal{S}_{curr}$. If unassigned entities remain, the 
controller fetches the next available structural template 
$\square_{stmt}$. Crucially, the search strategy operates 
purely at the macro-level: it selects an unassigned entity 
$e_i \in \mathcal{E} \setminus E_{curr}$ and binds it to 
the entity identifier slot 
($\square_{stmt}.\square_{entity} \gets e_i$). The 
concrete geometric instantiation of the remaining sub-holes 
( $\square_{type}$, $\square_{pos}$, and 
$\square_{heading}$) is then delegated to the micro-level 
SMT solver via the \textsc{SynthesizeStatement} function 
(detailed later in Alg. \ref{alg:synthesize_statement}). 
This strict decoupling of macro-level topological ordering 
from micro-level geometric constraint solving effectively 
bounds the search complexity, making an otherwise 
intractable problem solvable.

\begin{algorithm}[!t]
	\caption{\textsc{ConstraintSolvingBasedSynthesizer}($\sigma_{in}, \text{mode}$)}
	\label{alg:top_level}
	\setstretch{0.75}
	\begin{algorithmic}[1]
		\State \textbf{input:} A concrete input scene $\sigma_{in}$ containing a finite set of entities $\mathcal{E}$, and search strategy $\text{mode} \in \{\text{DFS, BFS}\}$.
		\State \textbf{output:} An optimal scenario program $\mathcal{P}^*$ such that $\mathcal{P}^* \models \sigma_{in}$.
		
		\State $\mathcal{S} \gets \textsc{SketchGenerator}(\sigma_{in})$  \label{sg}
		
		\If{$\text{mode} = \text{BFS}$}
			\State $\mathcal{P}^* \gets \textsc{SketchBasedSearchBFS}(\mathcal{S}, \sigma_{in})$\label{bfs}
			\Else
			\State $\mathcal{P}^* \gets \textsc{SketchBasedSearchDFS}(\mathcal{S}, \sigma_{in})$\label{dfs}
			
		\EndIf
		
		\State \textbf{return} $\mathcal{P}^*$
	\end{algorithmic}
\end{algorithm}

\begin{algorithm}[!t]
	\caption{\textsc{SketchBasedSearchBFS}($\mathcal{S}, 
	\sigma_{in}$)}
	\label{alg:search_controller_bfs}
	\setstretch{0.75}
	\begin{algorithmic}[1]
		\State \textbf{input:} The input sketch $\mathcal{S}$ and the concrete scene $\sigma_{in}$ containing a finite set of entities $\mathcal{E}$.
		\State \textbf{output:} The optimal program $\mathcal{P}^*$ such that $\mathcal{P}^* \models \sigma_{in}$.
		\State $\mathcal{P}^* \gets \perp, \ \mathcal{M}_{smt} \gets \{\}$ \hfill $\triangleright$ $\perp$ denotes null; $\mathcal{M}_{smt}$: SMT query cache \label{2l3}
		
		\State Let $e_{ego} \in \mathcal{E}$ be the designated ego entity \label{2l4}
		\State $Q \gets \{ (\mathcal{S}, 0) \}$ \label{2l5}
		
		\While{$Q \neq \emptyset$} \label{2l6}
		\State $\mathcal{T}_{layer} \gets \{\}$ \hfill $\triangleright$ Hash table: $E_{next} \mapsto (\mathcal{S}, C)$ for current layer\label{2l7}
		
		\For{\textbf{each} $(\mathcal{S}_{curr}, C_{curr}) \in Q$}\label{2l8}
		\State $E_{curr} \gets \text{Entities}(\mathcal{S}_{curr})$ \hfill $\triangleright$ Extract set of entities already declared in $\mathcal{S}_{curr}$ \label{2l9}
		
		\If{$ \mathcal{E} \setminus E_{curr} = \emptyset $}\label{2l10}
		\State \Comment{Optimal complete program found; $\mathcal{P}^* \models \sigma_{in}$ by construction}
		\State $ \mathcal{P}^* \gets \mathcal{S}_{curr}$\label{2l12}
		\EndIf
		
		\For{\textbf{each} $e_i \in \mathcal{E} \setminus E_{curr}$} \label{2l13}
		\State $\square_{stmt} \gets 
		\text{GetNextStatement}(\mathcal{S}_{curr})$ \hfill 
		$\triangleright$ Fetch next uninstantiated 
		statement hole \label{2l14}
		\State $\square_{stmt}.\square_{entity} \gets e_i$ 
		\hfill $\triangleright$ Dot notation accesses 
		sub-holes of the statement template \label{2l15}
		\State $E_{ref} \gets \{e_{ego}\} \cup E_{curr}$\label{2l16}
		\State $\mathcal{S}_{new}, \Delta c \gets 
		\textsc{SynthesizeStatement}(\square_{stmt}, 
		e_i, E_{ref}, \mathcal{S}_{curr}, \sigma_{in}, 
		\mathcal{M}_{smt})$\label{2l17}
		\State $C_{new} \gets C_{curr} + \Delta c$\label{2l18}
		
		\State $E_{next} \gets E_{curr} \cup 
		\{e_i\}$\label{2l19}
		\If{$E_{next} \notin \mathcal{T}_{layer}$ \textbf{ 
		or } $C_{new} < 
		\mathcal{T}_{layer}[E_{next}].C$}\label{2l21}
		\State \Comment{\textbf{Common Prefix Abstraction-based 
		Pruning}}
		\State $\mathcal{T}_{layer}[E_{next}] \gets (\mathcal{S}_{new}, C_{new})$\label{2l22}
		
		\EndIf
		\EndFor
		\EndFor
		\State $Q \gets \mathcal{T}_{layer}$ \label{2l23}
		\EndWhile
		
		\State \textbf{return} $\mathcal{P}^*$
	\end{algorithmic}
\end{algorithm}

\vspace{1ex}
\noindent\textbf{Breadth-First Search (BFS) and Common Prefix Abstraction-based Pruning.}
Alg. \ref{alg:search_controller_bfs} details the BFS approach to solve the synthesis problem.
The algorithm takes an initial sketch $\mathcal{S}$ and the concrete scene $\sigma_{in}$ (containing a finite set of entities $\mathcal{E}$) as inputs, aiming to synthesize the optimal program $\mathcal{P}^*$ with the minimal semantic penalty.
Initialization begins by setting the optimal program to null ($\perp$), and allocating an empty global cache $\mathcal{M}_{smt}$ for SMT solver memoization (Line \ref{2l3}). After designating the ego object $e_{ego}$ as the primary reference (Line \ref{2l4}), the active queue $Q$ is seeded with the initial sketch and a starting cost of zero (Line \ref{2l5}).
From this initial state, the algorithm systematically expands the search tree layer by layer using the active queue $Q$ (Lines \ref{2l6}-\ref{2l23}). 
Each layer corresponds to a specific depth of the sketch tree (\textit{i.e.}, the number of instantiated statements). 
While BFS demands a larger memory footprint to store the active frontier, its strictly layer-wise expansion naturally aligns all partial programs of identical depths. This structural uniformity makes it exceptionally well-suited for the pruning technique introduced below.

Specifically, different declaration orders of the same group of entities result in the identical set $E_{curr}$.
Because the context for future SMT solver depends solely on \textit{which} entities have been declared rather than the \textit{order} of their declaration, these distinct paths can be abstracted into an identical search state.
To exploit this, we introduce \textbf{Common Prefix Abstraction-based Pruning}. 
At the beginning of each depth iteration, we initialize an empty layer-wise optimal prefix table $\mathcal{T}_{layer}$ (Line \ref{2l7}). This table is a hash map that associates each set of declared entities $E$ with the lowest-cost partial sketch $(\mathcal{S}, C)$ that has declared exactly those entities.
During the expansion phase (Lines \ref{2l8}-\ref{2l22}), if a sketch has successfully instantiated all entities ($\mathcal{E}\setminus E_{curr} = \emptyset$), we evaluate it to update the global optimal program $\mathcal{P}^*$(Lines \ref{2l10}-\ref{2l12}). 
If the sketch remains incomplete, the algorithm branches out to explore all valid single-step extensions by iterating over the pool of unassigned entities $\mathcal{E} \setminus E_{curr}$ (Line \ref{2l13}).
For each candidate $e_i \in \mathcal{E} \setminus E_{curr}$, we bind it to the next statement hole and invoke the \textsc{SynthesizeStatement} procedure (using $\{e_{ego}\} \cup E_{curr}$ as references) to deduce the extended sketch $\mathcal{S}_{new}$ and its incremental penalty $\Delta c$(Line \ref{2l14}-\ref{2l17}).
After calculating the total accumulated cost $C_{new}$ and updating the defined entity set $E_{next}$ (Lines \ref{2l18}-\ref{2l19}), the algorithm enforces the Common Prefix Abstraction. Specifically, the unordered entity subset $E_{next}$ serves as a unique hash key to index the layer-wise optimal prefix table $\mathcal{T}_{layer}$. 
Using the unordered subset $E_{next}$ as a unique hash key, $\mathcal{T}_{layer}$ is updated only if the new path achieves a strictly lower cumulative cost $C_{new}$ (Lines \ref{2l21}-\ref{2l22}). 
Consequently, any sub-optimal permutations that reach the identical entity abstraction are mathematically dominated and seamlessly pruned. 
Finally, after fully expanding the current layer, all optimal surviving states stored in $\mathcal{T}_{layer}$ are extracted to form the new active queue $Q$ for the subsequent depth iteration (Line \ref{2l23}). This pruning mechanism drastically truncates the breadth of the search tree while strictly preserving global optimality.	
Ultimately, at the maximum depth, all valid paths instantiate the full entity set $\mathcal{E}$. Consequently, $\mathcal{T}_{layer}$ collapses to a single key $\mathcal{E}$, naturally isolating the globally optimal program $\mathcal{P}^*$ (Lines \ref{2l10}-\ref{2l12}) and terminating the search.

\vspace{1ex}
\noindent\textbf{Depth-First Search (DFS) and Branch-and-Bound.} 
\begin{algorithm}[t]
	\caption{\textsc{SketchBasedSearchDFS}($\mathcal{S}, 
	\sigma_{in}$)}
	\label{alg:search_controller}
	\setstretch{0.75}
	\begin{algorithmic}[1]
		\State \textbf{input:} The input sketch $\mathcal{S}$ and the concrete scene $\sigma_{in}$ containing a finite set of entities $\mathcal{E}$.
		\State \textbf{output:} The optimal program $\mathcal{P}^*$ such that $\mathcal{P}^* \models \sigma_{in}$.
		\State $\mathcal{P}^* \gets \perp, \ C^* \gets \infty, \ \mathcal{T}_{global} \gets \{\}, \ \mathcal{M}_{smt} \gets \{\}$ \hfill $\triangleright$ $\perp$ denotes null
		\State \Comment{$\mathcal{T}_{global}$: global hash table $E_{curr} \mapsto$ minimal cost; $\mathcal{M}_{smt}$: SMT query cache}
		
		\State Let $e_{ego} \in \mathcal{E}$ be the designated ego entity
		\State $Stack \gets [(\mathcal{S}, 0)]$ \label{3l5}
		
		\While{$Stack \neq \emptyset$}\label{3l6}
		\State $(\mathcal{S}_{curr}, C_{curr}) \gets Stack.pop()$\label{3l7}
		\State $E_{curr} \gets \text{Entities}(\mathcal{S}_{curr})$ \hfill $\triangleright$ Extract declared entities from $\mathcal{S}_{curr}$ \label{3l8}
		
		\If{$E_{curr} \in \mathcal{T}_{global}$ \textbf{ and } $C_{curr} \ge \mathcal{T}_{global}[E_{curr}]$}\label{3l9}
		\State \textbf{continue} \hfill $\triangleright$ 
		\textbf{Common 
		Prefix Abstraction-based Pruning}\label{3l10}
		\EndIf
		\State $\mathcal{T}_{global}[E_{curr}] \gets C_{curr}$\label{3l11}
		
		\If{$C_{curr} \ge C^*$} \hfill $\triangleright$ 
		\textbf{Branch 
		\& Bound Pruning}\label{3l12}
		\State \textbf{continue}\label{3l13}
		\EndIf
		
		\If{$ \mathcal{E} \setminus E_{curr} = \emptyset $}\label{3l14}
		\State $\mathcal{P}^* \gets \mathcal{S}_{curr}, \ C^* \gets C_{curr}$ \Comment{$\mathcal{P}^* \models \sigma_{in}$ by construction}\label{3l15}
		\State \textbf{continue}\label{3l16}
		\EndIf
		
		\For{\textbf{each} $e_i \in \mathcal{E} \setminus E_{curr}$} \label{3l17}
		\State $\square_{stmt} \gets 
		\text{GetNextStatement}(\mathcal{S}_{curr})$ \hfill 
		$\triangleright$ Fetch next uninstantiated 
		statement hole \label{3l18}
		
		\State $\square_{stmt}.\square_{entity} \gets e_i$ 
		\hfill $\triangleright$ Dot notation accesses 
		sub-holes of the statement template \label{3l19}
		\State $E_{ref} \gets \{e_{ego}\} \cup E_{curr}$\label{3l20}
		\State $\mathcal{S}_{new}, \Delta c \gets 
		\textsc{SynthesizeStatement}(\square_{stmt}, e_i, 
		E_{ref}, \mathcal{S}_{curr}, \sigma_{in}, 
		\mathcal{M}_{smt})$\label{3l21}
		
		\State $Stack.push((\mathcal{S}_{new}, C_{curr} + \Delta c))$ \label{3l22}
		\EndFor
		\EndWhile
		
		\State \textbf{return} $\mathcal{P}^*$
	\end{algorithmic}
\end{algorithm}
Alg. \ref{alg:search_controller} details the DFS variant. Sharing identical inputs and objectives with BFS, DFS utilizes a Last-In-First-Out (LIFO) \textit{Stack} to systematically explore the permutation space (Line \ref{3l5}). While standard DFS is renowned for its memory efficiency by only maintaining the active path, we introduce a global optimal prefix table $\mathcal{T}_{global}$ to support the aforementioned Common Prefix Abstraction across all explored paths, deliberately trading memory overhead for substantial search space reduction.

At each step, the algorithm pops a state and sequentially evaluates it against our dual pruning mechanisms. First, we adapt the \textbf{Common Prefix Abstraction-based Pruning} into the DFS framework to combat structural redundancies across divergent branches. If the traversal encounters an already visited state $E_{curr}$ but with a cost $C_{curr} \ge \mathcal{T}_{global}[E_{curr}]$, the branch is safely pruned as a sub-optimal prefix (Lines \ref{3l9}-\ref{3l10}). Otherwise, DFS explicitly updates $\mathcal{T}_{global}$ to continuously record the minimum cost required to reach this unordered entity subset signature (Line \ref{3l11}).
Following this, the state undergoes \textbf{Branch-and-Bound Pruning}. Because the semantic penalty increment is strictly non-negative ($\Delta c \ge 0$), the accumulated cost along any path is monotonically non-decreasing. Thus, if any partial sketch $\mathcal{S}_{curr}$ retrieved during traversal accrues a cost $C_{curr} \ge C^*$, its entire subsequent sub-tree is guaranteed to be sub-optimal and is immediately discarded (Lines \ref{3l12}-\ref{3l13}).
If the state survives both pruning checks and represents a fully synthesized scenario program ($\mathcal{E} \setminus E_{curr} = \emptyset$), the algorithm directly establishes or tightens the valid upper bound $C^*$ for the optimal cost (Lines \ref{3l14}-\ref{3l16}).
Finally, for incomplete states surviving all checks, the algorithm proceeds to the macro-micro expansion phase. Iterating through the unassigned entities, it binds each $e_i$ to a new statement hole and invokes the \textsc{SynthesizeStatement} procedure. The resulting extended sketches and their cumulatively updated costs are subsequently pushed onto the stack (Lines \ref{3l17}-\ref{3l22}). Together, these dual pruning strategies effectively constrain the combinatorial explosion, ensuring rapid convergence to the globally optimal program $\mathcal{P}^*$.

\vspace{-4pt}

\begin{lemma}\label{lemma1}
	Under the assumption of a perfect SMT oracle, we guarantee the soundness of our pruning strategies. Any partial program discarded during the search is either definitively suboptimal or safely redundant, ensuring that the discovery of a global optimal solution is never compromised.
\end{lemma}

\begin{proof}
	Our exploration involves two primary pruning heuristics: common prefix abstraction and branch-and-bound. We formalize the soundness (optimality-preserving property) of these strategies as follows:
	\vspace{-3pt}
	\begin{itemize}[leftmargin=10pt]
		\item \textbf{Soundness of Prefix Abstraction-based Pruning}: Pruning intermediate programs based on entity-set equivalence preserves the optimal solution.
		We define an equivalence relation over partial programs where two programs are equivalent if they declare the identical set of entities, irrespective of the declaration order. While the declaration permutation affects the current accumulated cost, appending identical subsequent entities to these equivalent programs will result in the same marginal cost increase. 
		Thus, within an equivalence class, strictly higher-cost programs are dominated. For cost ties, we simply retain the first explored program to eliminate redundancy.
		By safely discarding these dominated and redundant variants , the optimality of the final synthesis is mathematically guaranteed.

		\item  \textbf{Soundness of Branch-and-Bound Pruning}: Pruning partial programs whose cost exceeds the best-known upper bound preserves the optimal solution.
		Let $\mathcal{P}^*$ denote the fully synthesized program with the minimum cost discovered so far, acting as our global upper bound. As established in Definition \ref{costF}, the cost function $C(P)$ is monotonically \textit{non-decreasing} with respect to the addition of entities. If a partial program $P_{partial}$ currently under evaluation satisfies $C(P_{partial}) \ge C(\mathcal{P}^*)$, any complete extension $P'_{partial}$ derived from it will inherently satisfy $C(P'_{partial}) \ge C(P_{partial}) \ge C(\mathcal{P}^*)$. Thus, no sequence of subsequent derivations from $P_{partial}$ can yield a strictly better solution than $\mathcal{P}^*$. Pruning $P_{partial}$ eliminates only suboptimal sub-spaces or redundant ties, rendering the Branch-and-Bound Pruning strategy profoundly safe.
	\end{itemize}
\end{proof}

\vspace{-25pt}

\subsection{SMT-Based Statement Synthesizer}
\label{subsec:statement_synthesizer}

\begin{algorithm}[t]
	\caption{\textsc{SynthesizeStatement}($\square_{stmt}, 
	e_{new}, E_{ref}, \mathcal{S}_{curr}, \sigma_{in}, 
	\mathcal{M}_{smt}$)}
	\label{alg:synthesize_statement}
	\setstretch{0.85}
	\begin{algorithmic}[1]
		\State \textbf{input:} The assigned statement hole 
		$\square_{stmt}$, target entity $e_{new}$, 
		reference set $E_{ref}$, sketch 
		$\mathcal{S}_{curr}$, concrete scene $\sigma_{in}$, 
		and global SMT cache $\mathcal{M}_{smt}$.
		\State \textbf{output:} The updated sketch $\mathcal{S}_{new}$ and accumulated semantic penalty $\Delta c$.
		
		\State $\Delta c \gets 0$ \Comment{Phase 0: Resolve 
		Semantic Type 
		($\square_{stmt}.\square_{type}$)}\label{4l4}
		
		\State $\square_{stmt}.\square_{type} \gets 
		\sigma_{in}(e_{new}).type$ \label{4l5}
		
		\State $pos\_found \gets \text{False}$ 
		\Comment{Phase 1: Synthesize Positional Relation 
		($\square_{stmt}.\square_{pos}$)}\label{4l6}
		\For{\textbf{each} $e_{ref} \in E_{ref}$}\label{4l8}
		\For{\textbf{each} $rel \in \{\textbf{ahead of}, \textbf{behind}, \textbf{left of}, \textbf{right of}\}$}\label{4l9}
		\State $sat, vals \gets \textsc{SMTPosRelSolve}(rel, e_{new}, e_{ref}, \sigma_{in}, \mathcal{M}_{smt})$
		\State \hfill $\triangleright$ Queries $\mathcal{M}_{smt}[(rel, e_{new}, e_{ref})]$ first; on miss, calls solver and caches $(sat, vals)$
		\If{$sat$}\label{4l11}
		\State $\square_{stmt}.\square_{pos} \gets 
		rel \ e_{ref} \textbf{ by Range}(vals.a,vals.b)$ 
		\State $pos\_found \gets \text{True}$ \label{4l14}
		\EndIf
		\EndFor
		\If{$pos\_found$} \textbf{break} \EndIf
		\EndFor
		
		\If{\textbf{not} $pos\_found$} \label{4l16}
		\State $\Delta c \gets \Delta c + 1$
		\State $sat, vals \gets \textsc{SMTLocalXYSolve}(e_{new}, e_{ego}, \sigma_{in}, \mathcal{M}_{smt})$
		\State \hfill $\triangleright$ Queries $\mathcal{M}_{smt}[(\text{"local\_xy"}, e_{new}, e_{ego})]$; caches $(sat, vals)$ on miss
		\State $\square_{stmt}.\square_{pos} \gets 
		\textbf{on Range}(vals.a,vals.b) @ 
		\textbf{Range}(vals.c,vals.d) \textbf{ relative to 
		} e_{ego}$\label{4l19}
		\EndIf
		
		\State $head\_found \gets \text{False}$ 
		\Comment{Phase 2: Synthesize Heading Relation 
		($\square_{stmt}.\square_{heading}$)}\label{4l20}
		\For{\textbf{each} $e_{ref} \in E_{ref}$}
		\For{\textbf{each} $rel \in \{\textbf{toward}, \textbf{away\_from}\}$}
		\State $sat \gets \textsc{SMTHeadingRelSolve}(rel, e_{new}, e_{ref}, \sigma_{in}, \mathcal{M}_{smt})$
		\State \hfill $\triangleright$ Queries $\mathcal{M}_{smt}[(rel, e_{new}, e_{ref})]$; caches $sat$ on miss
		\If{$sat$}
		\State $\square_{stmt}.\square_{heading} \gets 
		\textbf{facing } rel \ e_{ref}$
		\State $head\_found \gets \text{True}$
		\EndIf
		\EndFor
		\If{$head\_found$} \textbf{break} \EndIf
		\EndFor
		
		\If{\textbf{not} $head\_found$} \label{4l30}
		\State $\Delta c \gets \Delta c + 1$
		\State $sat, vals \gets \textsc{SMTRelativeAngleSolve}(e_{new}, e_{ego}, \sigma_{in}, \mathcal{M}_{smt})$
		\State \hfill $\triangleright$ Queries $\mathcal{M}_{smt}[(\text{"relative\_heading"}, e_{new}, e_{ego})]$; caches $(sat, vals)$ on miss
		\State $\square_{stmt}.\square_{heading} \gets 
		\textbf{facing Range}(vals.a,vals.b) \textbf{ 
		relative to } e_{ego}$\label{4l33}
		\EndIf
		
		\State $stmt \gets 
		\square_{stmt}.\text{assemble}()$ \hfill 
		$\triangleright$ Construct complete statement from 
		resolved holes\label{4l34}
		\State $\mathcal{S}_{new} \gets \mathcal{S}_{curr} 
		\oplus \{ \square_{stmt} \mapsto stmt 
		\}$\label{4l35}
		\State \textbf{return} $\mathcal{S}_{new}, \Delta c$\label{4l36}
	\end{algorithmic}
\end{algorithm}

While the Sketch Based Search dictates the macro-level declaration order of entities, the exact geometric and semantic constraints are resolved by the SMT-Based \textsc{SynthesizeStatement}, as detailed in Alg. \ref{alg:synthesize_statement}. 
The primary objective of this module is to instantiate the 
sub-holes of a given statement template ($\square_{stmt}$) 
by validating candidate spatial relations against the 
physical ground truth provided in $\sigma_{in}$.

\begin{itemize}[leftmargin=10pt]
	\item \textbf{Semantic Type Resolution.} 
	Before invoking any geometric solvers, the synthesizer first deterministically extracts and binds the entity's semantic type directly from the concrete scene $\sigma_{in}$ (Line \ref{4l5}).
	\item \textbf{Synthesize Positional Relation.}
	For $\square_{pos}$ (Lines \ref{4l6}-\ref{4l19}), the module 
	evaluates $e_{new}$ against each reference $e_{ref} \in 
	E_{ref}$ using spatial primitives (\textit{e.g.}, 
	\textbf{ahead of}). It invokes \textsc{SMTPosRelSolve} to 
	verify geometric validity against the constraints in 
	Eq.~\ref{eq:pos_semantics}. To optimize performance, 
	$\mathcal{M}_{smt}$ acts as a memoization table implemented 
	as a hash table. Each SMT solver invocation uses a 
	three-element key: (query type, target entity, reference 
	entity). For example, positional relation queries use keys 
	like $(\text{"ahead"}, e_{new}, e_{ref})$, while fallback 
	solvers use keys like $(\text{"local\_xy"}, e_{new}, 
	e_{ego})$. The stored result is either a tuple $(sat, vals)$ 
	containing satisfiability and boundary values, or a boolean 
	$sat$ for high-level heading relation queries. Before 
	invoking the underlying solver, each function checks whether 
	the query exists in $\mathcal{M}_{smt}$ and reuses cached 
	results to bypass redundant solver calls across different 
	search branches. If \texttt{SAT}, $\square_{pos}$ is 
	instantiated with the resulting boundary intervals $vals$ 
	(Lines \ref{4l11}-\ref{4l14}). If no high-level relation is 
	valid, the synthesizer falls back to resolving coordinates 
	in the ego's local frame (Eq.~\ref{eq:local_xy_semantics}). 
	By Definition \ref{costF}, this lower abstraction level is 
	penalized by incrementing the semantic cost $\Delta c \gets 
	\Delta c + 1$ (Lines \ref{4l16}-\ref{4l19}).
	\item \textbf{Synthesize Heading Relation.} 
	For $\square_{heading}$ (Lines \ref{4l20}-\ref{4l33}), 
	the synthesizer follows a similar paradigm by verifying 
	constraints against Eq.~\ref{eq:heading_semantics}. It 
	first seeks high-level relations (\textit{e.g.}, 
	\textbf{toward}) via \textsc{SMTHeadingRelSolve}; if 
	unsuccessful, it defaults to a relative heading angle. 
	Following Definition \ref{costF}, this fallback incurs 
	a penalty $\Delta c \gets \Delta c + 1$ (Lines 
	\ref{4l30}-\ref{4l33}).
\end{itemize}

Upon resolving all sub-holes, the instantiated statement is integrated into $\mathcal{S}_{curr}$ using the state extension operator $\oplus$, yielding the next sketch $\mathcal{S}_{new}$ (Lines \ref{4l34}-\ref{4l35}). Finally, the algorithm returns $\mathcal{S}_{new}$ and the accumulated penalty $\Delta c$ to guide the macro-level pruning logic (Line \ref{4l36}).

\vspace{-3pt}
\subsection{Theoretical Analysis} 
\label{subsec:theoretical_analysis}
Before presenting the theoretical analysis, we assume the existence of a perfect SMT oracle. This idealized assumption ensures that the underlying solver is unconditionally sound and complete, allowing us to cleanly isolate and formally establish the theoretical guarantees of our algorithm.

\subsubsection{Soundness} 

In this section, we prove that the proposed scenario synthesis algorithm is sound with respect to the problem formulation (Definition \ref{Problem}). That is, if the algorithm successfully returns a scenario program, this program is guaranteed to not only geometrically encapsulate the input concrete scene, but also mathematically minimize the semantic penalty cost.

\begin{theorem}[Soundness]
	\label{thm:soundness}
	Given an input concrete scene $\sigma_{in}$ representing the deterministic state configuration of all $n$ entities, if Alg. \ref{alg:top_level} terminates and returns a scenario program $\mathcal{P}^*$, then $\mathcal{P}^*$ is the optimal consistent program. Formally, $\mathcal{P}^*$ satisfies both consistency and optimality:
	\begin{equation}
		\sigma_{in} \in \mathbb{D}(\mathcal{P}^*, \Omega_{\emptyset}) \quad \text{and} \quad \forall P \text{ s.t. } \sigma_{in} \in \mathbb{D}(P, \Omega_{\emptyset}), \ \mathcal{C}(\mathcal{P}^*) \le \mathcal{C}(P)
	\end{equation}
\end{theorem}

\begin{proof}
	The proof is bipartite: we first establish the consistency of the synthesized program, and subsequently prove its optimality.
	
	\textbf{Consistency ($\sigma_{in} \in \mathbb{D}(\mathcal{P}^*, \Omega_{\emptyset})$).}
	Let the synthesized program be a sequence of statements $\mathcal{P}^* = S_1\ S_2\ \dots\ S_n$, where each $S_i$ declares entity $e_i$. Let $\sigma_i$ denote the partial concrete scene restricting $\sigma_{in}$ to the first $i$ entities. We proceed by induction on the incremental construction of the scenario $\Omega_i = \mathbb{D}(S_1\ \dots\ S_i, \Omega_{\emptyset})$.
	
	\textit{Base Case:} Before any statement is evaluated, $\Omega_0 = \{\sigma_{\emptyset}\}$ and the 0-entity scene is $\sigma_0 = \sigma_{\emptyset}$. Thus, $\sigma_0 \in \Omega_0$ trivially holds.
	
	\textit{Inductive Step:} Assume that after synthesizing $i-1$ statements, the partial scene $\sigma_{i-1}$ is valid within the current scenario, \textit{i.e.}, $\sigma_{i-1} \in \Omega_{i-1}$. 
	For the $i$-th statement $S_i$ ($e_i = $\textbf{new}\ 
	$type_i\ pos_i, heading_i$), let the concrete state of 
	entity $e_i$ in the input scene be $\sigma_{in}(e_i) = 
	(t_i, p_i, \theta_i)$. Let $s_{pos} = 
	\sigma_{i-1}(ref_{pos})$ and $s_{heading} = 
	\sigma_{i-1}(ref_{heading})$ denote the states of the 
	reference entities specified in $pos_i$ and $heading_i$ 
	respectively (or undefined for absolute coordinates). 
	According to our relational statement semantics 
	$\mathbb{V}$, the updated partial scene $\sigma_i = 
	\sigma_{i-1} \cup \{e_i \mapsto (t_i, p_i, \theta_i)\}$ 
	belongs to $\Omega_i = \mathbb{V}(S_i, \Omega_{i-1})$ 
	if and only if:
	\begin{equation}
		t_i = \mathbb{T}(type_i) \ \land \ p_i \in 
		\mathbb{P}(pos_i, s_{pos}) \ \land \ \theta_i \in 
		\mathbb{H}(heading_i, s_{heading}, p_i)
	\end{equation}
	
	Alg. \ref{alg:synthesize_statement} invokes the SMT-based synthesizer to resolve these constraints, satisfying three conditions:
	\begin{itemize}[leftmargin=16pt]
		\item \textbf{Semantic Type ($t_i = \mathbb{T}(type_i)$):}
		The synthesizer directly extracts the exact entity type from $\sigma_{in}$; thus, the type equality strictly holds.

		\item \textbf{Position ($p_i \in \mathbb{P}(pos_i, 
		s_{pos})$):}
		The $pos_i$ is synthesized by querying the SMT 
		solver with the hypothesized positional relation. 
		Because the solver uses the exact coordinates of 
		the previously declared reference entities from 
		$\sigma_{in}$ (identically preserved in 
		$\sigma_{i-1}$) as known constants, a \texttt{SAT} 
		response guarantees that the synthesized boundaries 
		legally enclose $p_i$. Therefore, $p_i \in 
		\mathbb{P}(pos_i, s_{pos})$.

		\item \textbf{Heading ($\theta_i \in 
		\mathbb{H}(heading_i, s_{heading}, p_i)$):}
		Similarly, the heading relation $heading_i$ is 
		synthesized by encoding the geometric assertions in 
		Eq. \ref{eq:heading_semantics}. The solver utilizes 
		the reference entity's state and the newly 
		confirmed position $p_i$ as constants. The 
		instantiation of $heading_i$ only occurs upon a 
		\texttt{SAT} verification against the concrete 
		orientation $\theta_i$, ensuring $\theta_i \in 
		\mathbb{H}(heading_i, s_{heading}, p_i)$.
	\end{itemize}
	Since all conditions are satisfied, $\sigma_i \in \Omega_i$ holds for step $i$. By mathematical induction, upon evaluating the final statement $S_n$, the complete input scene $\sigma_n = \sigma_{in}$ is guaranteed to be a member of the final scenario $\Omega_n = \mathbb{D}(\mathcal{P}^*, \Omega_{\emptyset})$.
	
	\textbf{Optimality ($\mathcal{C}(\mathcal{P}^*) \le \mathcal{C}(P)$).}
	Having established the synthesized program $\mathcal{P}^*$ is consistent, we must prove it minimizes the cost function $\mathcal{C}$. Our search algorithm systematically explores the finite combinatorial space of entity declarations. Following Lemma \ref{lemma1}, our pruning strategies (Common Prefix Abstraction and Branch-and-Bound) are provably optimality-preserving. They exclusively discard infeasible, strictly suboptimal, or safely redundant partial programs, ensuring the globally optimal solution is never erroneously pruned. Consequently, upon termination, the consistent program $\mathcal{P}^*$ that the search identifies and returns strictly minimizes the semantic cost.
\end{proof}

\vspace{-2pt}
\subsubsection{Completeness}
In this section, we prove the completeness of the synthesis algorithm with respect to Definition \ref{Problem}. Our algorithm provides a strong guarantee: it will always terminate and successfully return the optimal consistent scenario program $\mathcal{P}^*$ for any given input scene.

\vspace{-2pt}
\begin{theorem}[Completeness]
	\label{thm:completeness}
	Given any concrete input scene $\sigma_{in}$ containing a finite set of $n$ entities, Alg. \ref{alg:top_level} is guaranteed to terminate and return the optimal consistent scenario program $\mathcal{P}^*$ as defined in Definition \ref{Problem}.
\end{theorem}
\vspace{-8pt}
\begin{proof}
	The proof relies on three foundational properties of the synthesis framework.
	
	\textbf{Existence Guarantee of a Consistent Program.}	
	Our DSL provides an absolute semantic cover: for any exact entity state $(t, p, \theta) \in \sigma_{in}$, it accommodates either high-level relative primitives or explicit numerical fallbacks that can formulate bounding regions encompassing any coordinate in $\mathbb{R}^2$ and angle in $[0, 360)$. Thus, the existence of at least one semantically consistent program is theoretically guaranteed. Given the discrete state space, a global cost minimum must inherently exist among these consistent candidates, ensuring the existence of the optimal target program $\mathcal{P}^*$.

	\textbf{Finite Abstraction of the Search Space.} 
	Although spatial parameters reside in the continuous real domain, our algorithm maintains a discrete search tree by delegating constraint resolution to the SMT solver. Assuming the solver acts as an idealized oracle returning satisfying assignments for consistent spatial primitives, the search complexity becomes strictly bounded by the finite number of entities and categorical relations, rendering the exploration finite and tractable.

	\textbf{Safety of Pruning Strategies.} 
	By Lemma \ref{lemma1}, our pruning strategy is rigorously safe. While aggressively evaluating semantic costs to discard strictly suboptimal or safely redundant branches, it is mathematically proven to never erroneously exhaust the space of consistent solutions. By strictly retaining at least one globally optimal pathway at every step, the strategy guarantees that a canonical optimal target program $\mathcal{P}^*$ remains persistently reachable within the search tree.
	
	In conclusion, because $\mathcal{P}^*$ inherently exists, the SMT delegation keeps the search space strictly finite, and the pruning guarantees the reachability of $\mathcal{P}^*$, the algorithm is complete. It will unconditionally terminate and return $\mathcal{P}^*$.
\end{proof}

\subsubsection{Deterministic Search Space Reduction in BFS}
Beyond soundness and completeness, the Common Prefix Abstraction yields a deterministic, strictly bounded state reduction specifically within the Breadth-First Search (BFS) framework, a theoretical guarantee absent in DFS.
\begin{theorem}[Complexity Bound of BFS with Common Prefix Abstraction]
	\label{thm:bfs_reduction}
	For a scene with $N = |\mathcal{E}|$ entities, the naive combinatorial search explores $O(N!)$ permutations. By applying Common Prefix Abstraction-based Pruning within layer-wise BFS, the maximum number of active states at depth $k$ is strictly bounded by $\binom{N}{k}$, reducing the overall search complexity (in terms of \textsc{SynthesizeStatement} invocations) to $O(N \cdot 2^{N-1})$.
\end{theorem}
\vspace{-8pt}

\begin{proof}
	At any search depth $k \in \{0, \dots, N\}$, a naive search maintains $P(N, k) = \frac{N!}{(N-k)!}$ active permutations. However, our abstraction defines states as equivalent if they share the same unordered subset of $k$ declared entities. Because BFS strictly expands layer-by-layer, it exhaustively evaluates all depth-$k$ permutations before proceeding. The prefix table $\mathcal{T}_{layer}$ deterministically collapses these into at most $\binom{N}{k}$ unique canonical states (the combinations of $N$ entities taken $k$ at a time), retaining only the optimal path for each subset.
	
	Consequently, when expanding from depth $k$ to $k+1$, 
	each of the bounded $\binom{N}{k}$ states attempts to 
	instantiate the remaining $(N-k)$ unassigned entities. 
	The maximum total number of 
	\textsc{SynthesizeStatement} invocations across the 
	entire search is exactly bounded by:
	$$\sum_{k=0}^{N-1} \binom{N}{k} (N - k) = N \cdot 
	2^{N-1}$$
	This mathematically proves a deterministic complexity reduction from a factorial bound $O(N!)$ to an exponential bound $O(N \cdot 2^{N-1})$.
\end{proof}

\vspace{-2mm}
\subsection{Discussion}
\label{subsec:discussion}

\subsubsection{Trade-offs Between BFS and DFS}
\label{subsubsec:tradeoffs}
While Theorem \ref{thm:bfs_reduction} establishes a strict mathematical bound for BFS, DFS offers complementary advantages, creating a distinct time-memory trade-off.

\textit{1) Time Complexity. } BFS guarantees deterministic branch reduction via Common Prefix Abstrac\-tion-based Pruning. Conversely, DFS is highly sensitive to exploration order; reaching a subset via a sub-optimal path early can trigger redundant sub-tree exploration. However, DFS offsets this theoretical vulnerability by rapidly diving to leaf nodes. This quickly establishes a tight global upper bound $C^*$, empowering aggressive Branch-and-Bound pruning for subsequent branches.

\textit{2) Space Complexity. } BFS isolates memory overhead to the current depth. Its queue $Q$ and layer-wise prefix table $\mathcal{T}_{layer}$ rebuild iteratively, peaking at $O(\binom{N}{\lfloor N/2 \rfloor})$. In contrast, while DFS features a lightweight $O(N)$ call stack, its global prefix table $\mathcal{T}_{global}$ must permanently cache visited states across all depths. Consequently, this global map grows monotonically, demanding up to $O(2^N)$ memory in the worst case.

\vspace{-3pt}
\subsubsection{Alternative Search Algorithms}
\label{subsubsec:other_algorithms}
In addition to BFS and DFS, other search algorithms can also be integrated into our framework. For instance, A* search~\cite{A-star}, a classical best-first search algorithm, can be combined with Prefix Abstraction-based Pruning to guide the exploration towards lower-cost programs. Similar to BFS, A* can leverage the prefix abstraction to merge equivalent states and reduce redundant exploration. Prior work has successfully applied A* to program synthesis with synthesis-specific heuristics~\cite{shah2020admissible, mell2024optimal}. We empirically evaluate and compare these alternatives in Section \ref{subsec:rq2_results}.

	%

\section{Evaluation}
This section evaluates the effectiveness and efficiency of our proposed method. 
\highlight{We compare our approach against several LLM-based baselines. For the efficiency evaluation, we conduct an internal ablation of B\&B and CPA within the DFS- and BFS-based configurations and compare these configurations with A*-based search\cite{A-star}.}
Our evaluation aims to answer the following research questions:
\begin{itemize}[leftmargin=10pt]
	\item \textbf{RQ1} How effective is our method in synthesizing scenario programs compared to existing approaches?
	\item \textbf{RQ2} \highlight{How efficient is our proposed algorithm, how effectively do the two pruning strategies improve the computational scalability, and how does it compare with A*-based search? }
\end{itemize}
\noindent\textbf{\textit{Benchmark}}. 
To address RQ1 and RQ2, we derive 20 scenario generation tasks from the annotated, real-world nuScenes \cite{nuscenes} dataset. To systematically evaluate performance across varying levels of difficulty, these scenes are structured to feature a strictly increasing number of entities, scaling incrementally from a relatively simple configuration of 5 entities up to a highly complex scene containing 24 entities.

\noindent\textbf{\textit{Implementation and Experimental Setup}}.
Our synthesis algorithm leverages dReal \cite{dReal} as the underlying constraint solver to handle nonlinear formulas over the reals. Specifically, dReal implements the framework of $\delta$-complete decision procedures \cite{DeltaComplete}. This allows the tool to systematically evaluate various nonlinear real functions, including the trigonometric functions that are essential for our spatial layout synthesis. Since the baseline approaches primarily rely on cloud-based LLM API invocations, their effectiveness evaluations (RQ1) are hardware-independent. 
To rigorously answer RQ2, all efficiency and ablation experiments were executed sequentially in isolated Docker containers running on a single Linux server (AMD RT PRO 7995WX CPU), with each explicitly constrained to 4 cores and 16GB of RAM.

\subsection{Results of RQ1}
\label{subsec:rq1_results}
To address RQ1 and evaluate our method's effectiveness, we compare our approach against the following three scenario program synthesis techniques.
\begin{itemize}[leftmargin=10pt]
	\item \textbf{ChatScene:} We select ChatScene \cite{ChatScene} as a baseline, a recent work that leverages LLMs and a pre-built knowledge base to synthesize Scenic code from natural language. 
	\highlight{Its direct assembly mode is designed around a pre-built
		knowledge base and therefore cannot directly cover unseen layouts that
		are absent from that knowledge base. We consequently use the alternative
		few-shot prompting mode provided by ChatScene, which retrieves relevant
		code examples while allowing new Scenic programs to be generated for the
		layouts in our benchmark.}
	
	\item \textbf{ScenicNL:} An existing work \cite{scenicNL} that utilizes a compositional prompting strategy to generate Scenic code from natural language. Specifically, it integrates various advanced techniques, such as Tree of Thoughts, few-shot prompting, constrained decoding, and a compiler-in-the-loop approach, to fully leverage the generation capabilities of LLMs.
	
	\item \textbf{DirectLLMGen:} We employ a multimodal LLM to directly synthesize scenario programs. The input context consists exclusively of a raw scene image, its corresponding label data file, and a supplementary image with rendered labels, deliberately omitting any textual description of the scene itself. Instead, our carefully designed natural language prompt is used strictly to enforce the syntax format of the target scenario language. This setup forces the LLM to rely on multimodal reasoning abilities to extract spatial layout from the visually augmented inputs, outputting the final program according to the predefined syntactic constraints.
\end{itemize}

To evaluate the LLM-based baselines, we utilize three frontier models: \texttt{GPT-4o}, \texttt{GPT-5.4}, and \texttt{Gemini-3.1-Pro}. Since ScenicNL's LMQL-based constrained decoding \cite{LMQL} requires token-level logit control unavailable in modern chat APIs, we adopt its alternative configuration\footnote{We utilize ScenicNL's \texttt{PREDICT\_TOT\_THEN\_HYDE} configuration (combining ToT and few-shot prompting). Additionally, we ported the authors' compiler-in-the-loop feedback mechanism into this pipeline.}. 
For our proposed method, we employ the BFS strategy equipped with Common Prefix Abstraction-based Pruning.


\highlight{In our evaluation, all methods are provided with the same
	nuScenes scene labels, including the entity types, positions, and
	headings. For \methodname, these labels directly constitute the
	structured concrete scene. DirectLLMGen receives the raw image together
	with the label file and an image with rendered labels, while ChatScene
	and ScenicNL receive textual descriptions containing the same label
	information. All LLM-based configurations are evaluated with at most
	five generation attempts per task: ChatScene and DirectLLMGen are run
	five times, and ScenicNL uses one initial generation followed by at most
	four compiler-feedback retries.}

Table \ref{tab:rq1_esr} presents our experimental results evaluated across syntax correctness, semantic correctness, and program cost. 
Regarding \textbf{syntax correctness}, all programs are strictly validated via the official Scenic compiler. Our method, DirectLLMGen, and ScenicNL achieve near-perfect syntax pass rates, whereas ChatScene exhibits high variance. Our symbolic synthesis fundamentally guarantees strict syntax by design. DirectLLMGen and ScenicNL ensure validity via strict prompt constraints and ToT with compiler-in-the-loop feedback, respectively. Conversely, lacking explicit formatting or feedback loops, ChatScene relies solely on few-shot examples, making its performance highly unstable and dependent on the LLM's prior knowledge.

\begin{table}[t]
	\scriptsize
	\centering
	\caption{Evaluation results of RQ1. The metrics in each cell represent \textbf{Syntax Correct / Semantic Correct / Average Program Cost} (Syn / Sem / Cost). For LLM-based methods, syntax success is achieved if at least one out of five attempts generates valid code. Semantic success is subsequently evaluated only on syntactically correct outputs, requiring at least one attempt to satisfy the semantic criteria. The average semantic cost is then computed exclusively over these semantically successful attempts. $\checkmark$ indicates correct, $\times$ indicates failure, and - indicates not applicable.}
	\label{tab:rq1_esr}
	\vspace{-2pt}
	\resizebox{\linewidth}{!}{
		\begin{tabular}{l ccc ccc ccc c}
			\toprule
			\multirow{2}{*}{\textbf{Task ID (\# Ent.)}} & \multicolumn{3}{c}{\textbf{DirectLLMGen}} & \multicolumn{3}{c}{\textbf{ChatScene}} & \multicolumn{3}{c}{\textbf{ScenicNL}} & \multirow{2}{*}{\textbf{Ours}} \\
			\cmidrule(lr){2-4} \cmidrule(lr){5-7} \cmidrule(lr){8-10}
			& \textbf{GPT-4o} & \textbf{GPT-5.4} & \textbf{Gemini} & \textbf{GPT-4o} & \textbf{GPT-5.4} & \textbf{Gemini} & \textbf{GPT-4o} & \textbf{GPT-5.4} & \textbf{Gemini} & \\
			\midrule
			Task 1 (5)   & $\checkmark$ / $\times$ / - & $\checkmark$ / $\times$ / - & $\checkmark$ / $\checkmark$ / 10.0 & $\times$ / - / - & $\checkmark$ / $\times$ / - & $\times$ / - / - & $\checkmark$ / $\checkmark$ / 20.0 & $\checkmark$ / $\checkmark$ / 20.0 & $\checkmark$ / $\checkmark$ / 20.0 & $\checkmark$ / $\checkmark$ / 0 \\
			Task 2 (6)   & $\checkmark$ / $\times$ / - & $\checkmark$ / $\times$ / - & $\checkmark$ / $\times$ / - & $\checkmark$ / $\times$ / - & $\checkmark$ / $\times$ / - & $\checkmark$ / $\times$ / - & $\checkmark$ / $\checkmark$ / 24.0 & $\checkmark$ / $\checkmark$ / 24.0 & $\checkmark$ / $\checkmark$ / 24.0 & $\checkmark$ / $\checkmark$ / 5.0 \\
			Task 3 (7)   & $\checkmark$ / $\times$ / - & $\checkmark$ / $\times$ / - & $\checkmark$ / $\times$ / - & $\times$ / - / - & $\checkmark$ / $\times$ / - & $\times$ / - / - & $\checkmark$ / $\checkmark$ / 28.0 & $\checkmark$ / $\checkmark$ / 28.0 & $\checkmark$ / $\checkmark$ / 28.0 & $\checkmark$ / $\checkmark$ / 1.0 \\
			Task 4 (8)   & $\checkmark$ / $\times$ / - & $\checkmark$ / $\times$ / - & $\checkmark$ / $\checkmark$ / 32.0 & $\times$ / - / - & $\checkmark$ / $\times$ / - & $\checkmark$ / $\times$ / - & $\checkmark$ / $\checkmark$ / 32.0 & $\checkmark$ / $\checkmark$ / 32.0 & $\checkmark$ / $\checkmark$ / 32.0 & $\checkmark$ / $\checkmark$ / 7.0 \\
			Task 5 (9)   & $\checkmark$ / $\times$ / - & $\checkmark$ / $\times$ / - & $\checkmark$ / $\checkmark$ / 36.0 & $\times$ / - / - & $\checkmark$ / $\times$ / - & $\checkmark$ / $\times$ / - & $\checkmark$ / $\checkmark$ / 36.0 & $\checkmark$ / $\checkmark$ / 36.0 & $\checkmark$ / $\checkmark$ / 36.0 & $\checkmark$ / $\checkmark$ / 3.0 \\
			Task 6 (10)  & $\checkmark$ / $\times$ / - & $\checkmark$ / $\times$ / - & $\checkmark$ / $\times$ / - & $\times$ / - / - & $\checkmark$ / $\times$ / - & $\checkmark$ / $\times$ / - & $\checkmark$ / $\checkmark$ / 40.0 & $\checkmark$ / $\checkmark$ / 40.0 & $\checkmark$ / $\checkmark$ / 40.0 & $\checkmark$ / $\checkmark$ / 1.0 \\
			Task 7 (11)  & $\checkmark$ / $\times$ / - & $\checkmark$ / $\times$ / - & $\checkmark$ / $\times$ / - & $\times$ / - / - & $\checkmark$ / $\times$ / - & $\checkmark$ / $\times$ / - & $\times$ / - / - & $\checkmark$ / $\checkmark$ / 44.0 & $\checkmark$ / $\checkmark$ / 44.0 & $\checkmark$ / $\checkmark$ / 8.0 \\
			Task 8 (12)  & $\checkmark$ / $\times$ / - & $\checkmark$ / $\times$ / - & $\checkmark$ / $\times$ / - & $\checkmark$ / $\times$ / - & $\checkmark$ / $\times$ / - & $\checkmark$ / $\times$ / - & $\checkmark$ / $\checkmark$ / 48.0 & $\checkmark$ / $\checkmark$ / 48.0 & $\checkmark$ / $\checkmark$ / 48.0 & $\checkmark$ / $\checkmark$ / 6.0 \\
			Task 9 (13)  & $\checkmark$ / $\times$ / - & $\checkmark$ / $\times$ / - & $\checkmark$ / $\times$ / - & $\times$ / - / - & $\checkmark$ / $\times$ / - & $\times$ / - / - & $\checkmark$ / $\checkmark$ / 52.0 & $\checkmark$ / $\checkmark$ / 52.0 & $\checkmark$ / $\checkmark$ / 52.0 & $\checkmark$ / $\checkmark$ / 9.0 \\
			Task 10 (14) & $\checkmark$ / $\times$ / - & $\checkmark$ / $\times$ / - & $\checkmark$ / $\checkmark$ / 56.0 & $\times$ / - / - & $\checkmark$ / $\times$ / - & $\times$ / - / - & $\checkmark$ / $\checkmark$ / 56.0 & $\checkmark$ / $\checkmark$ / 56.0 & $\checkmark$ / $\checkmark$ / 56.0 & $\checkmark$ / $\checkmark$ / 4.0 \\
			Task 11 (15) & $\checkmark$ / $\times$ / - & $\checkmark$ / $\times$ / - & $\checkmark$ / $\times$ / - & $\times$ / - / - & $\checkmark$ / $\times$ / - & $\times$ / - / - & $\checkmark$ / $\checkmark$ / 60.0 & $\checkmark$ / $\checkmark$ / 60.0 & $\checkmark$ / $\checkmark$ / 60.0 & $\checkmark$ / $\checkmark$ / 5.0 \\
			Task 12 (16) & $\checkmark$ / $\times$ / - & $\checkmark$ / $\times$ / - & $\checkmark$ / $\times$ / - & $\times$ / - / - & $\checkmark$ / $\times$ / - & $\times$ / - / - & $\checkmark$ / $\checkmark$ / 64.0 & $\checkmark$ / $\checkmark$ / 64.0 & $\checkmark$ / $\checkmark$ / 64.0 & $\checkmark$ / $\checkmark$ / 13.0 \\
			Task 13 (17) & $\checkmark$ / $\times$ / - & $\checkmark$ / $\times$ / - & $\checkmark$ / $\times$ / - & $\checkmark$ / $\times$ / - & $\checkmark$ / $\times$ / - & $\checkmark$ / $\times$ / - & $\checkmark$ / $\checkmark$ / 68.0 & $\checkmark$ / $\checkmark$ / 68.0 & $\checkmark$ / $\checkmark$ / 68.0 & $\checkmark$ / $\checkmark$ / 8.0 \\
			Task 14 (18) & $\checkmark$ / $\times$ / - & $\checkmark$ / $\times$ / - & $\checkmark$ / $\times$ / - & $\times$ / - / - & $\checkmark$ / $\times$ / - & $\checkmark$ / $\times$ / - & $\checkmark$ / $\checkmark$ / 72.0 & $\checkmark$ / $\checkmark$ / 72.0 & $\checkmark$ / $\checkmark$ / 72.0 & $\checkmark$ / $\checkmark$ / 7.0 \\
			Task 15 (19) & $\checkmark$ / $\times$ / - & $\checkmark$ / $\times$ / - & $\checkmark$ / $\times$ / - & $\times$ / - / - & $\checkmark$ / $\times$ / - & $\checkmark$ / $\times$ / - & $\checkmark$ / $\checkmark$ / 76.0 & $\checkmark$ / $\checkmark$ / 76.0 & $\checkmark$ / $\checkmark$ / 76.0 & $\checkmark$ / $\checkmark$ / 7.0 \\
			Task 16 (20) & $\checkmark$ / $\times$ / - & $\checkmark$ / $\times$ / - & $\checkmark$ / $\times$ / - & $\times$ / - / - & $\checkmark$ / $\times$ / - & $\times$ / - / - & $\checkmark$ / $\checkmark$ / 80.0 & $\checkmark$ / $\checkmark$ / 80.0 & $\checkmark$ / $\checkmark$ / 80.0 & $\checkmark$ / $\checkmark$ / 0.0 \\
			Task 17 (21) & $\checkmark$ / $\times$ / - & $\checkmark$ / $\times$ / - & $\checkmark$ / $\times$ / - & $\times$ / - / - & $\checkmark$ / $\times$ / - & $\times$ / - / - & $\checkmark$ / $\checkmark$ / 84.0 & $\checkmark$ / $\checkmark$ / 84.0 & $\checkmark$ / $\checkmark$ / 84.0 & $\checkmark$ / $\checkmark$ / 12.0 \\
			Task 18 (22) & $\checkmark$ / $\times$ / - & $\checkmark$ / $\times$ / - & $\checkmark$ / $\times$ / - & $\times$ / - / - & $\checkmark$ / $\times$ / - & $\times$ / - / - & $\checkmark$ / $\checkmark$ / 88.0 & $\checkmark$ / $\checkmark$ / 88.0 & $\checkmark$ / $\checkmark$ / 88.0 & $\checkmark$ / $\checkmark$ / 0.0 \\
			Task 19 (23) & $\checkmark$ / $\times$ / - & $\checkmark$ / $\times$ / - & $\checkmark$ / $\times$ / - & $\times$ / - / - & $\checkmark$ / $\times$ / - & $\checkmark$ / $\times$ / - & $\checkmark$ / $\checkmark$ / 92.0 & $\checkmark$ / $\checkmark$ / 92.0 & $\checkmark$ / $\checkmark$ / 92.0 & $\checkmark$ / $\checkmark$ / 10.0 \\
			Task 20 (24) & $\checkmark$ / $\times$ / - & $\checkmark$ / $\times$ / - & $\checkmark$ / $\times$ / - & $\times$ / - / - & $\checkmark$ / $\times$ / - & $\times$ / - / - & $\checkmark$ / $\checkmark$ / 96.0 & $\checkmark$ / $\checkmark$ / 96.0 & $\checkmark$ / $\checkmark$ / 96.0 & $\checkmark$ / $\checkmark$ / 7.0 \\
			\midrule
			\textbf{Overall} & 100.0\% / 0.0\% / - & 100.0\% / 0.0\% / - & 100.0\% / 20.0\% / 33.5 & 15.0\% / 0.0\% / - & 100.0\% / 0.0\% / - & 50.0\% / 0.0\% / - & 95.0\% / 95.0\% / 58.7 & 100.0\% / 100.0\% / 58.0 & 100.0\% / 100.0\% / 58.0 & \textbf{100.0\% / 100.0\% / 5.7} \\
			\bottomrule
		\end{tabular}
	}
\end{table}

For syntactically valid programs, we evaluate \textbf{semantic correctness} (defined as whether the generated program's distribution space accurately encompasses the input scene). We use an automated evaluation pipeline for our method and DirectLLMGen (restricted to a Scenic subset), while ChatScene and ScenicNL (using full syntax) require a two-step evaluation: an initial automated check to ensure the compiled entity count matches the input (as any mismatch strictly invalidates the scene), followed by expert assessment. Results show our method and ScenicNL achieve complete semantic correctness, whereas DirectLLMGen only occasionally succeeds using \texttt{Gemini-3.1-Pro}. 
Crucially, all baselines achieve this merely by transcribing absolute coordinates from the labels. Even the sole exception (DirectLLMGen with \texttt{Gemini-3.1-Pro} on Task 1) simply calculated local relative positions. 

To quantify this, we evaluate the \textbf{program cost} (Def. \ref{costF}). As theoretically established, our symbolic synthesis guarantees the minimum-cost program. 
\highlight{Under our cost function, a lower cost indicates that a program relies more on relational descriptions and less on absolute coordinates. The evaluated baselines frequently reproduce the input coordinates directly, tying their outputs to the original coordinate frame. In contrast, \methodname\ generates relational and re-samplable programs that preserve the encoded spatial layout while allowing controlled scenario variations. We discuss the practical implications of these programs in Section~\ref{sec:discussion}.}

\vspace{2ex}
\noindent \textbf{Impact of Natural Language Bias (User 
Study).} The RQ1 evaluation above used 
expert-crafted 
prompts with explicit absolute coordinates. However, 
real-world user descriptions are inherently abstract and 
relational. To evaluate this bias, 8 computer science 
graduate students were provided
\begin{wraptable}{rt}{0.55\linewidth}
	\vspace{-2pt} %
	\scriptsize
	\centering
	\caption{Comparison of generation performance between 
	Expert-crafted descriptions (Expert) and User-generated 
	descriptions (User). For Expert, results indicate 
	success or failure. For User, results represent the 
	average pass rate across 8 participants. `Syn' and 
	`Sem' denote Syntax and Semantic outcomes, 
	respectively.}
	\label{tab:user_study}
	\vspace{-2pt} 
	\resizebox{0.75\linewidth}{!}{
		\begin{tabular}{ll cc cc | cc cc}
			\toprule
			\multirow{3}{*}{\textbf{Task}} & 
			\multirow{3}{*}{\textbf{Model}} & 
			\multicolumn{4}{c|}{\textbf{ChatScene}} & 
			\multicolumn{4}{c}{\textbf{ScenicNL}} \\
			\cmidrule(lr){3-6} \cmidrule(lr){7-10}
			& & \multicolumn{2}{c}{\textbf{Expert}} & 
			\multicolumn{2}{c|}{\textbf{User}} & 
			\multicolumn{2}{c}{\textbf{Expert}} & 
			\multicolumn{2}{c}{\textbf{User}} \\
			\cmidrule(lr){3-4} \cmidrule(lr){5-6} 
			\cmidrule(lr){7-8} \cmidrule(lr){9-10}
			& & \textbf{Syn} & \textbf{Sem} & \textbf{Syn} 
			& \textbf{Sem} & \textbf{Syn} & \textbf{Sem} & 
			\textbf{Syn} & \textbf{Sem} \\
			\midrule
			\multirow{3}{*}{Task 1} 
			& \texttt{GPT-4o}   & $\times$ & - & 0.0\% & 
			0.0\% & $\checkmark$ & $\checkmark$ & 75.0\% & 
			0.0\% \\
			& \texttt{GPT-5.4}  & $\checkmark$ & $\times$ & 
			87.5\% & 0.0\% & $\checkmark$ & $\checkmark$ & 
			87.5\% & 0.0\% \\
			& \texttt{Gemini-3.1-Pro}   & $\times$ & - & 
			37.5\% & 0.0\% & $\checkmark$ & $\checkmark$ & 
			62.5\% & 0.0\% \\
			\midrule
			\multirow{3}{*}{Task 5} 
			& \texttt{GPT-4o}   & $\times$ & - & 12.5\% & 
			0.0\% & $\checkmark$ & $\checkmark$ & 50.0\% & 
			0.0\% \\
			& \texttt{GPT-5.4}  & $\checkmark$ & $\times$ & 
			100.0\% & 0.0\% & $\checkmark$ & $\checkmark$ & 
			75.0\% & 0.0\% \\
			& \texttt{Gemini-3.1-Pro}   & $\checkmark$ & 
			$\times$ & 37.5\% & 0.0\% & $\checkmark$ & 
			$\checkmark$ & 75.0\% & 0.0\% \\
			\midrule
			\multirow{3}{*}{Task 10} 
			& \texttt{GPT-4o}   & $\times$ & - & 12.5\% & 
			0.0\% & $\checkmark$ & $\checkmark$ & 50.0\% & 
			0.0\% \\
			& \texttt{GPT-5.4}  & $\checkmark$ & $\times$ & 
			75.0\% & 0.0\% & $\checkmark$ & $\checkmark$ & 
			87.5\% & 0.0\% \\
			& \texttt{Gemini-3.1-Pro}   & $\times$ & - & 
			62.5\% & 0.0\% & $\checkmark$ & $\checkmark$ & 
			87.5\% & 0.0\% \\
			\midrule
			\multirow{3}{*}{Task 15} 
			& \texttt{GPT-4o}   & $\times$ & - & 12.5\% & 
			0.0\% & $\checkmark$ & $\checkmark$ & 75.0\% & 
			0.0\% \\
			& \texttt{GPT-5.4}  & $\checkmark$ & $\times$ & 
			100.0\% & 0.0\% & $\checkmark$ & $\checkmark$ & 
			87.5\% & 0.0\% \\
			& \texttt{Gemini-3.1-Pro}   & $\checkmark$ & 
			$\times$ & 50.0\% & 0.0\% & $\checkmark$ & 
			$\checkmark$ & 87.5\% & 0.0\% \\
			\midrule
			\multirow{3}{*}{Task 20} 
			& \texttt{GPT-4o}   & $\times$ & - & 12.5\% & 
			0.0\% & $\checkmark$ & $\checkmark$ & 50.0\% & 
			0.0\% \\
			& \texttt{GPT-5.4}  & $\checkmark$ & $\times$ & 
			87.5\% & 0.0\% & $\checkmark$ & $\checkmark$ & 
			100.0\% & 0.0\% \\
			& \texttt{Gemini-3.1-Pro}   & $\times$ & - & 
			25.0\% & 0.0\% & $\checkmark$ & $\checkmark$ & 
			62.5\% & 0.0\% \\
			\bottomrule
		\end{tabular}
	}
	\vspace{-2pt}
\end{wraptable}   
with raw scene images and 
corresponding labels, and asked to describe the spatial 
layouts of 5 representative scenes in natural language 
without any structural constraints.

As Table \ref{tab:user_study} shows, ChatScene's syntax generation remains highly model-dependent. Notably, ScenicNL's syntax success rate declined significantly; lacking direct coordinates to copy, it struggled to generate valid complex structures. 
Crucially, the semantic performance of both baselines deteriorated sharply. 
\highlight{This shows that, under the evaluated user-generated
	descriptions, the LLM-based configurations did not generate programs	
	that correctly captured the corresponding spatial layouts.}


\vspace{1ex}
\noindent\fbox{%
	\parbox{\dimexpr\linewidth-2\fboxsep-2\fboxrule\relax}{%
		\highlight{\textbf{Answer to RQ1.} Under the evaluated inputs,
			\methodname\ achieves complete syntax and semantic correctness on
			all 20 scenes, whereas the LLM-based baselines attain such
			correctness only by transcribing absolute coordinates. Under our
			cost function, \methodname\ produces substantially lower-cost, and therefore more relational and re-samplable programs than the
			baselines.}
	}%
}

\subsection{Results of RQ2}
\label{subsec:rq2_results}

\highlight{RQ2 evaluates the efficiency and scalability of our synthesis
	algorithms from two perspectives. We first conduct an ablation
	study to quantify the effects of Branch-and-Bound (B\&B) and Common
	Prefix Abstraction-based Pruning (CPA) on the DFS- and BFS-based
	configurations. We then compare these configurations with A* search as an alternative to explore the same synthesis space.}

\highlight{We measure execution time and peak memory consumption across
	the 20 tasks. The evaluation includes nine configurations: one
	monolithic SMT baseline, six DFS/BFS configurations for the internal
	ablation, and two A*-based comparison configurations.}
\vspace*{-4pt}
\begin{enumerate}[label=\arabic*), leftmargin=22pt]
	\item \textbf{Pure-SMT}: Our custom baseline encoding both combinatorial permutations and geometric constraints into a single monolithic SMT formula.
	\item \textbf{DFS-None}: A naive Depth-First Search baseline without any pruning.
	\item \textbf{BFS-None}: A naive Breadth-First Search baseline without any pruning.
	\item \textbf{DFS-B\&B}: DFS with only B\&B.
	\item \textbf{DFS-CPA}: DFS with only CPA.
	\item \textbf{DFS-B\&B-CPA}: DFS equipped with both B\&B and CPA.
	\item \textbf{BFS-CPA}: BFS equipped with CPA.
	\item \highlight{\textbf{A*}: An A* search baseline without CPA.}
	\item \highlight{\textbf{A*+CPA}: A hybrid comparison configuration
		that combines A* with CPA.}
\end{enumerate}
\vspace*{-4pt}

\highlight{\textbf{A*-based comparison.} Prior work has applied
	A* to program synthesis using synthesis-specific heuristics. Shah et
	al.~\cite{shah2020admissible} derive an approximately admissible neural
	heuristic from a continuous relaxation to guide A* and
	iterative-deepening search, while Mell et
	al.~\cite{mell2024optimal} derive an abstract-interpretation-based
	heuristic that bounds the best achievable objective within subtrees of
	partial programs. Motivated by these studies, we instantiate A* for our
	synthesis space using $f=g+h$, where $g$ is the accumulated cost of the
	current partial program and $h$ is a relaxed lower bound on the cost of
	declaring the remaining entities. For each undeclared entity, the
	relaxation permits references to any entity in the scene and uses the
	minimum possible declaration cost. Because these relaxed choices include
	all references available to the actual search, $h$ cannot overestimate
	the true remaining cost and is admissible. Pure A* retains different
	declaration orders as distinct states, whereas A*+CPA merges equivalent
	prefixes with the same declared entity set.}

\begin{table*}[t]
	\centering
	\scriptsize
	\caption{Detailed Performance across All Tasks ($N=5$ to $24$), where $N$ denotes the number of entities within the scene, reflecting the increasing complexity of the scenarios. All reported metrics are averaged over 5 independent runs. TO = Timeout ($>$ 7200s), OOM = Out of Memory ($>$ 8GB).}
	\vspace{-10pt}
	\label{tab:all_tasks}
	\renewcommand{\arraystretch}{1.0}
	\resizebox{\textwidth}{!}{
		\setlength{\tabcolsep}{2pt}
		\begin{tabular}{@{}c|c|rr|rr|rr|rr|rr|rr|rr|rr|rr@{}}
			\toprule
			\multirow{2}{*}{\textbf{Task}} & \multirow{2}{*}{$N$} & \multicolumn{2}{c|}{\textbf{Pure-SMT}} & \multicolumn{2}{c|}{\textbf{DFS-None}} & \multicolumn{2}{c|}{\textbf{BFS-None}} & \multicolumn{2}{c|}{\textbf{DFS-B\&B}} & \multicolumn{2}{c|}{\textbf{DFS-CPA}} & \multicolumn{2}{c|}{\textbf{DFS-B\&B-CPA}} & \multicolumn{2}{c|}{\textbf{BFS-CPA}} & \multicolumn{2}{c|}{\highlight{\textbf{A*}}} & \multicolumn{2}{c}{\highlight{\textbf{A*+CPA}}} \\
			& & Time (s) & Mem (MB) & Time (s) & Mem (MB) & Time (s) & Mem (MB) & Time (s) & Mem (MB) & Time (s) & Mem (MB) & Time (s) & Mem (MB) & Time (s) & Mem (MB) & \highlight{Time (s)} & \highlight{Mem (MB)} & \highlight{Time (s)} & \highlight{Mem (MB)} \\
			\midrule
			Task 1 & 5 & TO & 1449.9 & 0.47 & 43.7 & 0.44 & 43.7 & 0.43 & 43.8 & 0.50 & 46.0 & \textbf{0.40} & 47.6 & 0.47 & 43.6 & \highlight{0.42} & \highlight{\textbf{9.63}} & \highlight{0.42} & \highlight{9.66} \\
			Task 2 & 6 & TO & 2700.1 & 1.09 & 45.1 & 0.93 & 44.6 & 1.16 & 50.5 & 1.11 & 48.5 & 1.01 & 49.7 & 1.02 & 44.4 & \highlight{\textbf{0.90}} & \highlight{\textbf{13.99}} & \highlight{0.93} & \highlight{14.11} \\
			Task 3 & 7 & TO & 4799.7 & 1.55 & 47.3 & 1.41 & 49.8 & 1.21 & 49.1 & 1.47 & 49.4 & 1.17 & 49.5 & 1.40 & 46.3 & \highlight{\textbf{1.10}} & \highlight{\textbf{14.01}} & \highlight{1.15} & \highlight{14.13} \\
			Task 4 & 8 & TO & 7801.1 & 4.13 & 52.0 & 4.15 & 77.6 & 3.55 & 49.8 & 3.11 & 53.4 & 3.01 & 53.9 & 2.82 & 53.5 & \highlight{2.67} & \highlight{17.73} & \highlight{\textbf{2.46}} & \highlight{\textbf{14.19}} \\
			Task 5 & 9 & 5029.83 & OOM & 8.73 & 50.4 & 13.61 & 533.4 & 5.86 & 47.7 & 1.82 & 53.2 & 1.73 & 51.3 & 1.77 & 49.2 & \highlight{3.88} & \highlight{71.58} & \highlight{\textbf{1.43}} & \highlight{\textbf{14.20}} \\
			Task 6 & 10 & 6236.34 & OOM & 63.46 & 51.2 & 111.11 & 5005.9 & 3.28 & 49.0 & 3.10 & 56.1 & 2.74 & 55.9 & 2.91 & 51.7 & \highlight{3.10} & \highlight{46.10} & \highlight{\textbf{1.91}} & \highlight{\textbf{14.26}} \\
			Task 7 & 11 & 6083.55 & OOM & 1223.93 & 54.1 & 219.19 & OOM & 416.44 & 51.3 & 5.62 & 57.2 & 5.73 & 58.5 & 5.41 & 53.3 & \highlight{49.28} & \highlight{1129.74} & \highlight{\textbf{4.74}} & \highlight{\textbf{14.68}} \\
			Task 8 & 12 & 5392.32 & OOM & TO & 53.1 & 188.82 & OOM & 607.34 & 54.5 & 7.30 & 57.4 & 6.55 & 58.3 & 6.73 & 54.2 & \highlight{336.55} & \highlight{7106.95} & \highlight{\textbf{5.83}} & \highlight{\textbf{15.64}} \\
			Task 9 & 13 & 3580.38 & OOM & TO & 52.7 & 209.74 & OOM & TO & 53.8 & 9.09 & 58.5 & 8.24 & 59.0 & 7.89 & 54.5 & \highlight{OOM} & \highlight{OOM} & \highlight{\textbf{6.71}} & \highlight{\textbf{19.83}} \\
			Task 10 & 14 & 307.64 & OOM & TO & 54.5 & 163.66 & OOM & 1498.79 & 53.2 & 12.30 & 63.0 & 9.34 & 58.7 & 10.82 & 54.8 & \highlight{OOM} & \highlight{OOM} & \highlight{\textbf{7.93}} & \highlight{\textbf{16.37}} \\
			Task 11 & 15 & 8.03 & OOM & TO & 53.7 & 170.63 & OOM & TO & 52.5 & 15.44 & 71.7 & 10.65 & 66.9 & 11.93 & 61.8 & \highlight{245.85} & \highlight{5152.39} & \highlight{\textbf{7.03}} & \highlight{\textbf{17.97}} \\
			Task 12 & 16 & 8.03 & OOM & TO & 53.8 & 166.62 & OOM & TO & 52.8 & 23.81 & 102.6 & 23.77 & 103.1 & 15.39 & \textbf{74.7} & \highlight{OOM} & \highlight{OOM} & \highlight{\textbf{12.12}} & \highlight{84.70} \\
			Task 13 & 17 & 8.68 & OOM & TO & 54.3 & 169.90 & OOM & TO & 54.2 & 58.06 & 155.6 & 43.44 & 150.2 & 26.24 & 100.8 & \highlight{OOM} & \highlight{OOM} & \highlight{\textbf{13.91}} & \highlight{\textbf{99.63}} \\
			Task 14 & 18 & 9.03 & OOM & TO & 54.6 & 142.23 & OOM & TO & 54.4 & 89.31 & 261.5 & 72.53 & 254.2 & 39.20 & \textbf{152.2} & \highlight{OOM} & \highlight{OOM} & \highlight{\textbf{15.56}} & \highlight{153.83} \\
			Task 15 & 19 & 9.44 & OOM & TO & 52.6 & 143.92 & OOM & TO & 53.4 & 169.83 & 480.5 & 97.75 & 420.4 & 72.49 & 258.1 & \highlight{OOM} & \highlight{OOM} & \highlight{\textbf{13.28}} & \highlight{\textbf{125.81}} \\
			Task 16 & 20 & 10.05 & OOM & TO & 53.3 & 126.19 & OOM & 5.31 & \textbf{53.8} & 374.19 & 886.2 & \textbf{5.15} & 63.2 & 101.36 & 438.1 & \highlight{OOM} & \highlight{OOM} & \highlight{51.52} & \highlight{862.39} \\
			Task 17 & 21 & 10.24 & OOM & TO & 51.4 & 133.70 & OOM & TO & 53.3 & 615.63 & 1735.3 & 615.73 & 1736.0 & 279.37 & \textbf{874.8} & \highlight{OOM} & \highlight{OOM} & \highlight{\textbf{238.07}} & \highlight{2472.01} \\
			Task 18 & 22 & 11.09 & OOM & TO & 52.4 & 131.46 & OOM & 6.68 & \textbf{52.7} & 1775.63 & 3286.0 & \textbf{6.15} & 62.5 & 445.23 & 1593.1 & \highlight{OOM} & \highlight{OOM} & \highlight{213.80} & \highlight{3666.35} \\
			Task 19 & 23 & 12.44 & OOM & TO & 54.2 & 137.77 & OOM & TO & 53.2 & 3849.99 & 6704.4 & 3461.34 & 6681.6 & 1145.22 & \textbf{3140.4} & \highlight{OOM} & \highlight{OOM} & \highlight{\textbf{650.49}} & \highlight{7408.19} \\
			Task 20 & 24 & 13.13 & OOM & TO & 54.1 & 141.69 & OOM & TO & 54.7 & 3630.59 & OOM & 1851.33 & OOM & \textbf{2515.80} & \textbf{6129.6} & \highlight{OOM} & \highlight{OOM} & \highlight{OOM} & \highlight{OOM} \\
			\bottomrule
		\end{tabular}
	}
	\vspace{-3pt}
\end{table*}

\highlight{Table~\ref{tab:all_tasks} reports the execution time and peak memory
consumption of all nine configurations. We analyze the results in terms
of the necessity of decoupled synthesis, the effectiveness of CPA after
decoupling, and the time--memory trade-offs among the search strategies.}

\textbf{\highlight{1) Necessity of Decoupled Synthesis.}}
\highlight{The monolithic Pure-SMT baseline jointly encodes entity
	permutations and geometric constraints in a single formula and fails on
	all evaluated tasks, timing out for $N\leq 8$ and reaching the memory
	limit for $N\geq 9$. In contrast, after separating declaration-order
	search from SMT-based geometric reasoning, even DFS-None and BFS-None
	complete several small tasks within seconds. This contrast demonstrates
	that decoupling the discrete and geometric parts is essential for making
	the synthesis problem tractable. However, the unpruned decoupled
	configurations still fail as the number of entities grows, indicating
	that decoupling alone does not resolve the combinatorial search space.}

\textbf{\highlight{2) Effectiveness of CPA after Decoupling.}}
\highlight{Within the decoupled architecture, the main bottleneck
	is the repeated exploration of declaration orders. DFS-None
	times out from Task~8 onward, and BFS-None reaches the memory limit from
	Task~7 onward. Pure A* changes the expansion order and is more
	competitive on smaller tasks, but it still retains declaration orders
	and reaches the memory limit on most tasks from
	$N=13$ onward. CPA directly addresses this redundancy. DFS+B\&B remains
	sensitive to branch order and times out on many tasks, whereas
	DFS+B\&B-CPA completes every task through $N=23$. The A* comparison
	shows the same effect: at $N=12$, A*+CPA reduces execution time from
	336.55\,s to 5.83\,s and peak memory from 7,106.95\,MB to 15.64\,MB, and
	also completes every task through $N=23$. These results show that CPA
	effectively removes declaration-order redundancy across all three
	search orders: DFS, BFS, and A*.}

\textbf{\highlight{3) Time--Memory Trade-offs.}}
\highlight{DFS+B\&B-CPA can be exceptionally fast when depth-first
	search finds a low-cost complete program early, as on $N=20$ and
	$N=22$, where it finishes in 5.15\,s and 6.15\,s, respectively. BFS-CPA
	avoids a persistent global prefix table by releasing states layer by
	layer and is the only configuration that completes $N=24$ within the
	memory limit. A*+CPA also achieves strong runtime performance: it is the
	fastest completed configuration for $N=16$--$19$ and $N=23$, and it is
	consistently faster than BFS-CPA from $N=16$ through $N=23$. This speed
	comes with higher memory consumption at the largest scales. From
	$N=20$ to $N=23$, A*+CPA uses approximately two to three times the peak
	memory of BFS-CPA; at $N=23$, the two configurations consume
	7,408.19\,MB and 3,140.40\,MB, respectively. A*+CPA then reaches the
	memory limit at $N=24$, whereas BFS-CPA completes the task. These results
	show that A*+CPA often favors execution time, BFS-CPA provides stronger
	memory efficiency at the largest scales, and DFS+B\&B-CPA benefits most
	when an effective upper bound is established early.}
\highlight{To further investigate how the distribution of complete-program
	costs in the search space relates to search performance, we conduct a
	supplementary analysis (Appendix~\ref{app:score_distribution}). Since
	\methodname\ searches for a low-cost program, we summarize each scene's
	cost distribution by its mean program cost, where a lower mean means that more
	low-cost programs exist. The runtime of DFS+B\&B-CPA varies
	widely across scenes and is markedly shorter on scenes whose cost
	distribution admits more low-cost programs, whereas A*+CPA and BFS-CPA
	show no clear trend with the mean cost.}

\vspace{1ex}
\noindent\fbox{%
	\parbox{\dimexpr\linewidth-2\fboxsep-2\fboxrule\relax}{%
\highlight{\textbf{Answer to RQ2:}  Pure-SMT confirms the necessity of decoupling, while the
	results across DFS, BFS, and A* show that CPA is critical for removing
	declaration-order redundancy. DFS+B\&B-CPA benefits from early upper
	bounds, BFS-CPA provides the strongest memory scalability and is the
	only configuration that completes $N=24$, and A*+CPA is competitive in
	runtime but consumes more memory at larger scales.}
	}%
}

\vspace{2pt}
\subsection{Threats to Validity}
\highlight{The RQ1 comparison is subject to differences in task
	formulation and input modality. ChatScene and ScenicNL were designed to
	translate natural-language descriptions into Scenic, whereas
	\methodname\ synthesizes a relational program from a structured concrete
	scene. Providing all methods with the same quantitative spatial facts
	enables comparison of their generated programs, but does not reproduce
	the native end-to-end use case of each system. In addition, program cost
	measures the preference for relational descriptions defined by our cost
	function; it is not a general measure of program quality or language
	understanding. Therefore, the RQ1 results should be interpreted as an
	output-level comparison of syntax validity, consistency with the supplied
	scene, and relational abstraction under the evaluated inputs.}

	%
	

\vspace{-3ex}
\highlight{\section{Discussion}
	\label{sec:discussion}}
\highlight{\subsection{Practical Implications and Benefits}}

\highlight{The practical benefit of \methodname\ is that it transforms a
	structured concrete scene into a relational and re-samplable Scenic
	program. Replaying absolute coordinates reproduces only one concrete
	scene, whereas the synthesized program can generate multiple scenario
	variations while preserving the spatial relations encoded from the input
	scene. The input scene itself remains encompassed by the program, while
	relational primitives and \texttt{Range} expressions allow controlled
	variations in entity positions and headings. RQ2 further shows that such
	programs can be synthesized for the evaluated scenes containing up to 24
	entities.}
	
	\highlight{These scenario variations can support synthetic data
		generation and the targeted expansion of rare or difficult scenes.
		Johnson-Roberson et al.~\cite{johnson2017driving} showed that
		synthetic images generated in a virtual world could train a vehicle
		detector that transferred to real-world images, and Fremont
		et al.~\cite{Scenic} showed that Scenic-generated partial-occlusion
		scenes improved detector performance on difficult cases when added to
		training data. However, the Scenic programs used by Fremont et al.
		were hand-authored, while some existing Scenic generation methods
		start from natural-language descriptions. \methodname\ takes a
		different path: it synthesizes relational programs directly from
		real-world structured scenes. Our preliminary evaluation
		(Appendix~\ref{app:vlm_frontend}) further shows that these structured
		scenes can potentially be obtained automatically from unlabeled images.
		For users, this means providing a single real-world image
		suffices; no manual program writing or spatial description is needed.}
		
\vspace{-2pt}	
\highlight{\subsection{Integration with Perception, VLMs, and NL-based Methods}}

\highlight{We further consider how \methodname\ can be integrated with
	perception systems, VLMs, and NL-based methods, and preliminarily test
	the feasibility of the image-based frontend. In
	such an integrated workflow, a detector or VLM extracts the
	entity types, positions, and headings from images or videos, forming the
	structured concrete scene required by \methodname. 
	\methodname\ then uses the scene as the synthesis
		specification, applying SMT to check candidate spatial relations and
		synthesize their numerical ranges while searching for a low-cost
		relational Scenic program.}

\highlight{NL-based methods can also assist the symbolic search after a
	structured concrete scene is available. Pairwise position and heading
	relations can be derived from its labels and expressed as stylized NL,
	from which an LLM can propose a declaration order and candidate reference
	relations as a high-level sketch. Such guidance may be particularly
	useful for DFS, because a promising declaration order may
	produce a low-cost complete program earlier, tighten the global upper
	bound, and enable stronger pruning. Our preliminary experiment below
	evaluates the image-based frontend; implementing and evaluating this
	NL-guided search strategy remains future work.}

\highlight{The preliminary image-to-structured-scene evaluation covers
		50 images from the nuScenes validation split; the complete protocol and
		results are reported in Appendix~\ref{app:vlm_frontend}. Both GPT-5.5
		and Claude opus-4.8 show limited accuracy in recovering vehicle
		positions and headings. The domain-aligned PGD
		detector~\cite{wang2022probabilistic} trained on nuScenes performs
		substantially better, while the same architecture trained on KITTI
		degrades sharply, confirming that the advantage of a perception model
		depends critically on domain alignment. These results indicate that a
		domain-aligned perception model is the most promising source of the
		structured scene required by \methodname, while current general-purpose
		VLMs are not yet accurate enough for direct use.}
	

\section{Related Work}
While many computer vision approaches prioritize \textit{visual realism} (\textit{e.g.}, textures), we focus on \textit{structural realism}, which ensures physical plausibility for real-world scenarios. Consequently, appearance-centric generation is beyond our scope. We categorize relevant literature into: (1) Data-Driven Scenario Generation, (2) LLM-based Scenario Generation, and (3) Sketch-based Program Synthesis.

\textbf{Data-Driven Scenario Generation}. 
Recent data-driven works leverage neural architectures like VAEs \cite{aevs}, GANs \cite{GAN}, Diffusion Models \cite{Diffusion}, and Autoregressive models to learn spatial distributions. These methods typically synthesize scenes through four intermediate representations \cite{Survey}: explicit parameters \cite{TrafficGen,SceneControl,DiffuScene,SceneGen}, scene graphs \cite{MetaSim,MetaSim2}, implicit latent spaces \cite{XCube,Director3D}, and semantic layouts \cite{CityDreamer,Layout2Scene}. 
While representations like scene graphs and explicit parameters successfully preserve the high-level semantics of scenario topologies (\textit{e.g.}, categorical relationships and rough spatial layouts), they inherently leave out low-level geometric precision. In contrast, our scenario program-based approach perfectly bridges this gap, enabling mathematically precise editability over scenario structure. 
Furthermore, unlike data-driven methods that risk generating physically invalid scenarios, our symbolic paradigm utilizes SMT as a rigorous constraint solver to ensure physical and semantic validity. 
This precise constraint capability also enables high-fidelity real-world scene reconstruction, such as traffic accident restoration, aiding accident analysis and safety validation.

\textbf{LLM-based Scenario Generation.}
Leveraging the prior knowledge of LLMs, recent approaches broadly follow two paradigms. The first generates intermediate structural representations like scene graphs \cite{AnyHome, LLplace, GraphDreamer, DIScene}. To achieve more precise editability, the second paradigm employs LLMs for procedural code generation. Some works generate Python-like scripts \cite{3DGPT, SceneCraft, TheSceneLanguage}, while others synthesize domain-specific programs \cite{OpenUniverseIndoorSceneGeneration, ProceduralScenePrograms}. 
Notably, works like ChatScene \cite{ChatScene} and ScenicNL \cite{scenicNL}  target Scenic. While they excel at generating diverse scenarios with dynamic behaviors from open-vocabulary inputs, they struggle to capture accurate spatial relationships in static scenes. Consequently, they easily produce physically impossible configurations, causing critical issues for downstream tasks. 
In contrast, our approach grounds the synthesis process directly in real-world scene data. By leveraging an SMT-based symbolic framework, we rigorously extract and encode precise geometric relationships from physical environments into scenario programs, thereby guaranteeing structural realism.

\textbf{Sketch-based Program Synthesis}.
Sketch-based synthesis is a classic paradigm pioneered by Solar-Lezama et al. \cite{solar2009sketching}, which allows programmers to write partial programs with "\textit{holes}" to reduce the search space. Recent works advance both the automated generation and efficient solving of sketches \cite{chen2023data,zhang2025multi,chen2020multi,wang2017program}. Specifically, generation methods extract sketches from natural language or input-output examples \cite{chen2023data,zhang2025multi,chen2020multi}, while solving techniques predominantly employ enumerative search to instantiate holes \cite{gulwani2017program, li2024guiding, lee2021combining}. To improve efficiency, abstraction techniques are often leveraged to verify intermediate states and prune infeasible spaces \cite{wang2017program,zhang2025multi,guo2019program,liang2011scaling,wang2018learning}.
In contrast, our approach directly generates template-based sketches using semantic information from \highlight{real world scenes}. For solving, we combine enumerative search \cite{barke2020just} with SMT solving: the search phase instantiates discrete high-level spatial relations, while the SMT solver verifies their geometric feasibility and synthesizes the corresponding continuous numerical parameters, significantly accelerating the overall process.

	%

\vspace{-6pt}
\section{Conclusion and Future Work}
In this paper, we presented \methodname\ to address the sim-to-real gap in scenario generation. Instead of relying on monolithic constraint solving (which we showed to be intractable even for simple cases), we demonstrated that a decoupled search-and-solve architecture can efficiently navigate the vast permutation space of spatial relationships. By pruning the search tree, our approach successfully scales to highly dense environments with up to 24 entities.

To scale to even larger scenes, future work will investigate more advanced
heuristic pruning techniques. Furthermore, we plan to extend our framework's
expressiveness by supporting more complex syntactic constructs and
incorporating dynamic agent behaviors. \highlight{Finally, we plan to
integrate perception models and natural-language-guided search toward an
end-to-end pipeline from images to Scenic programs.}

\section*{Data Availability Statement}
The artifact associated with this paper is available at \url{https://doi.org/10.5281/zenodo.22807615}.


\section*{Acknowledgments}
This research was supported by National Key R\&D Program of China (No. 2024YFF0908003) and the NSFC Program (No. 62172429). We also thank the anonymous reviewers for their valuable feedback.
	
\vspace{-1ex}
\appendix
\highlight{\section{Preliminary Image-to-Structured-Scene Evaluation}
	\label{app:vlm_frontend}}

\begin{table*}[h]
	\centering
	\small
	\caption{\highlight{Vehicle detection and pose estimation on 50 nuScenes
			validation images (average 6.0 vehicles per image within 50\,m). TP,
			FP, and FN are per-image averages; position and heading errors are
			medians over matched pairs.}}
	\vspace{-5pt}
	\label{tab:image_vehicle_pose}
	\resizebox{\textwidth}{!}{
		\begin{tabular}{lrrrrrrrr}
			\toprule
			\textbf{Method} & \textbf{Precision} & \textbf{Recall} &
			\textbf{F1} & \textbf{TP/img} & \textbf{FP/img} & \textbf{FN/img} &
			\textbf{Pos.\,Err.(m)} & \textbf{Head\,Err.($^\circ$)} \\
			\midrule
			GPT-5.5          & 0.359 & 0.404 & 0.362 & 2.0 & 4.0 & 4.1 & 2.59 & 32.3 \\
			Claude opus-4.8  & 0.260 & 0.258 & 0.236 & 1.2 & 3.7 & 4.9 & 3.12 & 50.1 \\
			PGD-nuScenes     & 0.832 & 0.449 & 0.548 & 2.5 & 0.4 & 3.6 & 0.89 &  4.5 \\
			PGD-KITTI        & 0.089 & 0.184 & 0.101 & 0.6 & 6.0 & 5.5 & 3.07 & 96.9 \\
			\bottomrule
	\end{tabular}}
	\vspace{-5pt}
\end{table*}

\highlight{R2SGen takes a structured concrete scene as input. This
	experiment asks whether current vision-language models (VLMs) and
	perception models can produce such a scene directly from images, \textit{i.e.},
	whether the upstream frontend of the R2SGen pipeline can be automated.
	Given an image, each frontend must detect all vehicles within 50\,m of
	the ego vehicle and estimate their $(x, y)$ positions and headings in
	an ego-centric frame ($+x$ right, $+y$ forward, heading
	counter-clockwise from $+y$). We use 50 images from the nuScenes validation split,
	containing on average 6.0 vehicles per image within 50\,m. We use the
	validation split to ensure that PGD-nuScenes was not trained on any of
	the evaluated images.}

\highlight{Predicted and ground-truth vehicles are paired by optimal
	one-to-one assignment (the Hungarian algorithm \cite{hungarian1955,hungarian1957}): each prediction is
	matched to at most one ground-truth vehicle so as to minimize the total
	matching distance, and pairs farther than 5\,m are rejected. A
	\textbf{true positive (TP)} is a prediction correctly matched to a
	ground-truth vehicle within 5\,m; a \textbf{false positive (FP)} is an
	unmatched prediction (a hallucinated vehicle that does not exist in the
	scene); a \textbf{false negative (FN)} is an unmatched ground-truth
	vehicle (a real vehicle that was missed).
	\textbf{Precision} $=\mathrm{TP}/(\mathrm{TP}+\mathrm{FP})$ is the
	fraction of predictions that are correct; \textbf{recall}
	$=\mathrm{TP}/(\mathrm{TP}+\mathrm{FN})$ is the fraction of ground-truth
	vehicles that are detected; \textbf{F1}
	$=2\, \times \mathrm{TP}/(2\, \times \mathrm{TP}+\mathrm{FP}+\mathrm{FN})$ combines both
	into a single score from 0 to 1 that penalizes hallucinations and missed
	detections equally. For matched pairs, we report the median
	\textbf{position error} (Euclidean distance, in meters) and the median
	\textbf{heading error} (circular angular difference, in degrees).}

\highlight{We evaluate four configurations. GPT-5.5 and Claude opus-4.8
	are two general-purpose VLMs, each given the same task-specific prompt
	(defining the ego-centric frame and requesting vehicle positions and
	headings within 50\,m); each VLM is run three times per image, and the
	per-image metrics are averaged across the three runs. The other
	two use the Probabilistic and Geometric Depth (PGD) monocular 3D
	detector~\cite{wang2022probabilistic}: PGD-nuScenes uses a checkpoint
	trained on nuScenes (the same domain as the test images), while
	PGD-KITTI uses the same architecture trained on KITTI and applied to
	nuScenes images without adaptation. Both PGD configurations use a
	confidence threshold of 0.1.}

\highlight{Table~\ref{tab:image_vehicle_pose} shows that PGD-nuScenes
	achieves the best performance across all metrics, with the highest F1
	(0.548) and precision (0.832) and the lowest position and heading
	errors (0.89\,m and $4.5^\circ$). We note that while PGD-nuScenes
	achieves high precision (0.832), its recall remains moderate (0.449);
	a further analysis shows that recall drops to only 0.21
	for ground-truth vehicles beyond 30\,m, reflecting the fundamental
	difficulty of monocular depth estimation at long range. The two VLMs
	are substantially worse:
	GPT-5.5 (F1 $= 0.362$, position error $2.59$\,m, heading error
	$32.3^\circ$) outperforms Claude opus-4.8 (F1 $= 0.236$, $3.12$\,m,
	$50.1^\circ$). PGD-KITTI, using the same architecture trained on a
	different domain, collapses to F1 $= 0.101$ with a heading error of
	$96.9^\circ$, showing that the advantage of a perception model depends
	on domain alignment.}

\highlight{These results suggest that a domain-aligned perception model is
	the most promising source of structured scenes for \methodname, while
	current general-purpose VLMs are not yet accurate enough for direct
	use.}

\highlight{\section{Program-Cost Distribution and Search Behavior}
	\label{app:score_distribution}}

\highlight{We investigate how the distribution of complete-program
	costs in the search space relates to the performance of DFS+B\&B-CPA,
	BFS-CPA, and A*+CPA. We use 100 controlled scenes, each containing 18
	non-ego vehicles and a fixed ego vehicle, that cover a broad range of
	program-cost distributions. To summarize each scene's cost distribution,
	we use its mean program cost---a lower mean indicates that more low-cost
	programs exist.}

\highlight{Since the cost of declaring an entity depends only on the set
	of entities already declared (the property exploited by our Common Prefix
	Abstraction pruning), we can use a dynamic program to obtain the number
	of declaration orders at each total cost without enumerating all $N!$
	complete orders.}

\begin{figure*}[t]
	\centering
	\includegraphics[width=0.98\textwidth]
	{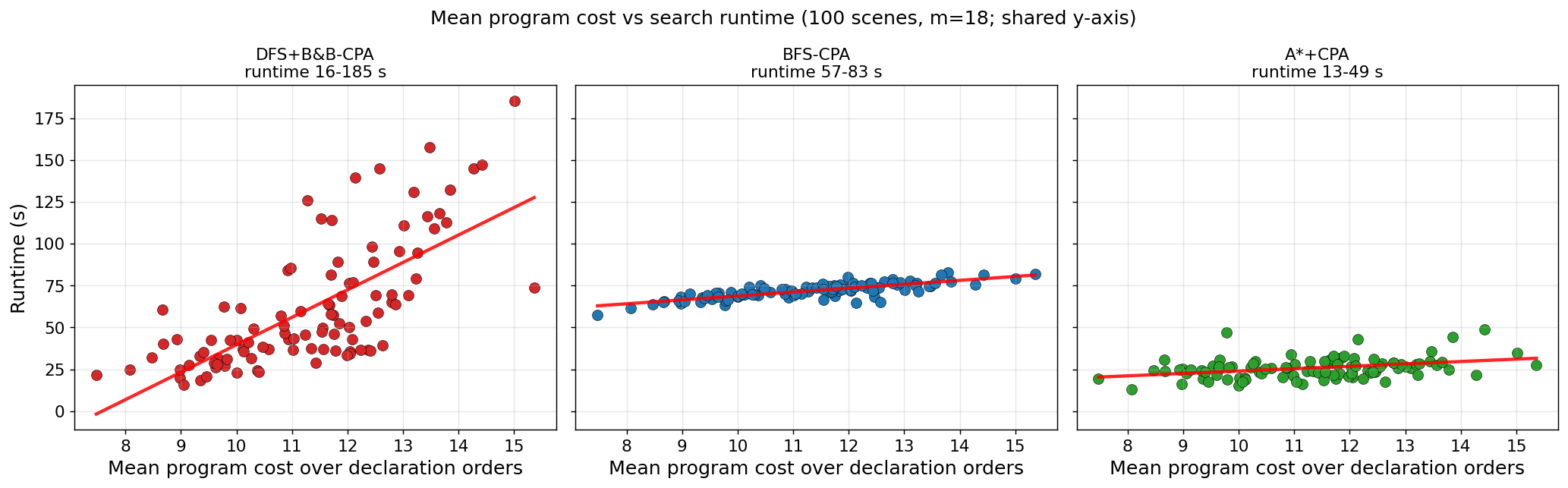}
	\caption{\highlight{Mean program cost versus search runtime over 100
			controlled scenes ($N=18$). The three panels show DFS+B\&B-CPA,
			BFS-CPA, and A*+CPA; each title reports the runtime range.}}
	\label{fig:score_distribution_scatter}
\end{figure*}

\highlight{We evaluate DFS+B\&B-CPA, BFS-CPA, and A*+CPA on each of the
	100 scenes.}
\highlight{DFS+B\&B-CPA's runtime varies widely across the 100 scenes,
	ranging from 16 to 185\,s, and tends to increase with the mean cost.
	We speculate that scenes admitting more low-cost programs let the
	search reach a low-cost complete program earlier and thereby tighten
	the branch-and-bound upper bound sooner. Neither A*+CPA
	($13$--$49$\,s) nor BFS-CPA ($57$--$83$\,s) shows a clear dependence
	on the mean cost.}
	
\vspace{1cm}

	\bibliographystyle{ACM-Reference-Format}
	
	\bibliography{sample-base}

\end{document}